\documentclass{article}

\usepackage[margin=1in]{geometry}
\usepackage{hyperref}
\usepackage{graphicx}
\usepackage{amssymb}
\usepackage{amsmath}
\usepackage{amsthm}
\usepackage[misc]{ifsym}
\usepackage[dvipsnames]{xcolor}
\usepackage[shortcuts]{extdash}
\usepackage[misc]{ifsym}
\usepackage{authblk}

\newtheorem{theorem}{Theorem}

\begin{document}

\title{Obfuscation of Images via Differential Privacy: From Facial Images to General Images \thanks{We gratefully acknowledge the financial support from the Natural Sciences and Engineering Research Council of Canada (NSERC) under Grants No. RGPIN-2020-06482, No. RGPIN-2016-06253 and No. CGSD2-503941-2017.}
}

\author[1]{William L. Croft}
\author[1]{J{\"o}rg-R{\"u}diger Sack}
\author[2]{Wei Shi}

\affil[1]{Carleton University, School of Computer Science, Ottawa, Canada}
\affil[2]{Carleton University, School of Information Technology, Ottawa, Canada}

\maketitle

\begin{abstract}
Due to the pervasiveness of image capturing devices in every-day life, images of individuals are routinely captured. Although this has enabled many benefits, it also infringes on personal privacy. A promising direction in research on obfuscation of facial images has been the work in the \emph{$k$-same} family of methods which employ the concept of \emph{$k$-anonymity} from database privacy. However, there are a number of deficiencies of $k$-anonymity that carry over to the $k$-same methods, detracting from their usefulness in practice. In this paper, we first outline several of these deficiencies and discuss their implications in the context of facial obfuscation. We then develop a framework through which we obtain a formal \emph{differentially private} guarantee for the obfuscation of facial images in generative machine learning models. Our approach provides a provable privacy guarantee that is not susceptible to the outlined deficiencies of $k$-same obfuscation and produces photo-realistic obfuscated output. In addition, we demonstrate through experimental comparisons that our approach can achieve comparable utility to $k$-same obfuscation in terms of preservation of useful features in the images. Furthermore, we propose a method to achieve differential privacy for any image (i.e., without restriction to facial images) through the direct modification of pixel intensities. Although the addition of noise to pixel intensities does not provide the high visual quality obtained via generative machine learning models, it offers greater versatility by eliminating the need for a trained model. We demonstrate that our proposed use of the \emph{exponential mechanism} in this context is able to provide superior visual quality to pixel-space obfuscation using the \emph{Laplace mechanism}.

\end{abstract}


\section{Introduction}

With the ever expanding presence of devices used to capture photos and video, visual privacy has become increasingly important. Images and video frames containing faces are routinely captured, e.g., through cameras, closed-circuit television systems \cite{38}, visual sensor networks \cite{37} and a host of other devices and methods \cite{57}. These systems have many benefits including mitigation of crime \cite{38,37}, improved care in assisted-living \cite{39}, and useful services such as Google Street View \cite{41}. However, despite the benefits of the legitimate applications, the potential for infringement on personal privacy must be taken seriously. 

Although many systems require only visual monitoring of behaviour, identities are often captured as well \cite{37}. In some areas, the degree of public surveillance is reaching levels where it becomes possible to profile and track much of the population \cite{39}. In cases where visual information is disseminated to the public, such as with Google Street View, it is imperative to hide the identities of individuals before the images are published \cite{42}. Failure to sufficiently protect privacy may allow undesirable inferences to be drawn about individuals or enable malicious activities such as voyeurism or stalking. Users of mobile devices have also expressed strong aversion to the collection of images from their mobile devices via the applications they use \cite{34}. Even in scenarios where users willingly share images to online platforms, they have expressed concerns over who is able to view their images \cite{35}. In a similar context, the privacy of individuals captured accidentally or without authorization in the backgrounds of images uploaded to such platforms should be taken into consideration. Rich visual information from these sources combined with the great advances in machine learning approaches to facial recognition (e.g., VGGFace \cite{19}) make the exploitation of unprotected visual data a relatively easy task. While such machine learning algorithms are no doubt beneficial in many contexts, it is essential for privacy protection approaches to be resistant to them.

To protect the privacy of individuals, methods for hiding identity via manipulation of the data can be employed. Many methods focus on obfuscation of facial identity since the face is often the most identifiable piece of information \cite{39,57}. In this context, facial identity refers to a visual representation of an individual's face which may potentially be exploited for unique identification of the individual. Trivially, the face could be covered by a uniformly coloured rectangle. This destroys all information about the face, guaranteeing that it can no longer be exploited to reveal an identity. However, this also destroys a great deal of utility and visual appeal. In scenarios where images are shared in online platforms, users have expressed a strong aversion to the use of such rectangles with respect to the visual quality and information content of the images \cite{32,33}.

Less severe methods of obfuscation present trade-offs between the level of privacy attainable and the utility of the data. It is well-known that methods of privacy protection necessitate such a trade-off \cite{59,60}. For example, a study on privacy-preserving systems used for visual awareness of co-workers in remote locations found that the level of privacy protection applied to hide identities directly impacted the value of the data to those using the systems \cite{43}. Preservation of utility is especially important for machine learning and data mining. Visual data can be used to learn about customers in retail environments \cite{49,50} and to detect anomalous or illegal events \cite{51,52}. It is therefore essential for a good method of obfuscation to preserve as much of the non-sensitive information as possible.

The exact specification of how utility is defined and what information is considered non-sensitive depends upon the context in which obfuscation is employed. For example, if data mining is to be employed for retail analytics, customer identity must be protected, however, certain demographic information may be deemed acceptable for release. This may lead to an interpretation of utility based on the degree of preservation of demographic information required for a data mining algorithm. In the case of images shared on social media, utility may instead be focused on a measure of visual aesthetic to ensure that obfuscation of bystanders in photos does not detract from the user experience. Regardless of the scenario, the goal of facial obfuscation is to provide an acceptable trade-off between protection of facial identity and preservation of non-sensitive information beneficial for utility in the obfuscated output.

A number of research directions have been explored for the obfuscation of visual identity in images, e.g., pixelization, blurring, etc. \cite{39}. While many approaches lack a formal privacy guarantee, the $k$-same \cite{1} family of approaches has gained a great deal of traction, largely thanks to its guarantee that for a chosen privacy parameter $k$, obfuscated individuals are indistinguishable within a group of $k$ potential true identities. Although this privacy guarantee is appealing, it suffers from susceptibilities (e.g., composition attacks \cite{36}) carried over from the disclosure control method of $k$-anonymity on which it is based. In this paper, we outline these susceptibilities in the context of privacy in images and propose an alternative, based on differential privacy, that addresses these susceptibilities.

Note that the task of facial obfuscation, which we address in this work, differs from privacy-preserving machine learning performed on facial images. Although both have the goal of preventing the leakage of sensitive information on facial identity, the former releases obfuscated facial images while the later releases a model trained on facial images. Facial obfuscation is a more general goal which can be applied in various ways. The obfuscated images may be used to enable privacy-preserving machine learning, may be shared publicly online, or may be used in other settings that require the protection of facial identity.

\subsection{Contributions and Paper Outline}

Our contributions in this work are as follows:

\begin{itemize}
	
	\item We examine susceptibilities of $k$-same obfuscation to composition attacks and inferences using background knowledge. We discuss theoretically and demonstrate empirically how the privacy guarantee can be violated. We discuss the implications this has on privacy in images and examine certain difficulties in the practical usage of $k$-same obfuscation.
	
	\item We propose the use of the formal privacy guarantee of differential privacy as a means to address the deficiencies of $k$-same obfuscation. We develop the first framework to apply differential privacy for the obfuscation of facial identity in images via generative machine learning models.
	
	\item We additionally propose a method to enforce differential privacy via the direct modification of pixel intensities. By giving up the high visual quality provided by generative models, this allows for a much more versatile approach that can obfuscate any image. We employ a process guided by an image quality function in order to improve the visual quality over the existing results for differential privacy in pixel-space.
	
	\item We conduct a series of experiments to compare differential privacy to $k$-same obfuscation on two well-known datasets for facial images. We demonstrate the resilience of differential privacy to composition and parrot attacks. Furthermore, the results of our experiments suggest that differential privacy offers a comparable level of utility in the obfuscated images to $k$-same obfuscation while providing a stronger privacy guarantee.
	
	\item We provide recommendations for the implementation of generative models for facial obfuscation as well as for image obfuscation in pixel-space to obtain effective and practical privacy protection.
	
\end{itemize}

We provide a review of existing work on the obfuscation of facial images in Section \ref{sec:lit_rev}. We then cover the deficiencies of $k$-same in Section \ref{sec:deficiences} and lay out a framework for differentially private obfuscation of images in Section \ref{sec:differential_privacy}. In Section \ref{sec:pixel}, we propose a method to obfuscate images in pixel-space using the exponential mechanism. Finally, we describe our experiments, provide comparisons between the methods of obfuscation and give an analysis of the results in Section \ref{sec:experiments}.

\section{Literature Review} \label{sec:lit_rev}

\subsection{Facial Obfuscation}

Perhaps the most well-known and earliest studied alterations to images for the prevention of human recognition of faces are pixelization \cite{28} and blurring \cite{27}. Pixelization decreases the information conveyed in an image by dividing the image into a grid of cells and setting all pixels within each cell to a common pixel intensity. Blurring involves the addition of, typically Gaussian, noise to the image. While these methods have been successful at foiling human recognition, they have been shown to be highly ineffective against machine recognition \cite{1}.

Other ad hoc methods of privacy protection involving variations on blurring \cite{12}, warping \cite{13}, morphing \cite{14} and face swapping \cite{8,6} have been studied. However, the methods that have gained the most momentum are those which offer a formal guarantee of privacy. This trend has been reinforced by the legal and legislative demands in the broader context of the release of sensitive data \cite{47,48,57}. To this end, $k$-same approaches have been quite successful. These approaches use an adaptation of $k$-anonymity \cite{31}, a concept from the field of database privacy that guarantees that an anonymized database record is linkable to at least $k$ possible identities. The first adaptation of this concept to image obfuscation worked by aligning a set of input images on their facial features, partitioning the set into clusters of $k$ or more similar images, and then averaging the pixels within each cluster to produce an averaged face which would replace each of the original faces in the cluster \cite{1}. By releasing only the averaged faces, it could be guaranteed that neither human nor machine recognition could do better than identifying the cluster of identities that produced the image, thus limiting the probability of successful re-identification by an upper bound of $\frac{1}{k}$.

One issue with the original $k$-same averaging of pixels was poor visual quality due to inexact alignment of facial features, leading to superimposed features. The $k$-same-m \cite{2} approach improved upon this by using an active appearance model (AAM) \cite{29} to obfuscate faces. AAMs are generative machine learning models for the approximation of visual representations of a particular class of objects (e.g., human faces). A model is trained on a set of images in order to learn about visual patterns and minimize differences with respect to shape and texture between the original images and the generated output of the model. The $k$-same-m approach first trains an AAM and then performs the clustering and averaging process within the parameter space of the model representations of faces to be obfuscated, thus eliminating the issue of superimposed features. Subsequent variants of $k$-same obfuscation via AAMs have since been proposed such as $k$-same-furthest \cite{58} which expressly selects clusters of images having highly dissimilar model parameters in order to make the re-identification of the clustered identities more challenging.

More recently, generative neural networks (NNs) have been applied for $k$-same obfuscation \cite{3,4}. Generative NNs are machine learning models that have shown great success in the generation of visual representations of input class labels \cite{26}. A generative NN passes the input labels through a sequence of convolutional layers, transforming them into features of finer granularity at each layer until reaching a pixel-space output. A training process adjusts weights used by filters in each convolutional layer in order to learn feature representations that minimize a loss function measuring the quality of the output. When trained on a set of images using identities as class labels, a generative NN is able to produce a visual approximation of an identity based on an input class vector. By providing input vectors in which $k$ identities are specified, the generative NN produces $k$-anonymous output.

Efforts have also been devoted to the preservation of utility in the obfuscated images. The $k$-same-select approach \cite{9} proposed partitioning the input images into classes based on the information to be preserved (e.g., male and female identities). By running a separate clustering process within each partition, images within each cluster would share the same class, thus preserving this information in the averaged version. This idea has been extended to the $k$-same-m model by training a different AAM for each combination over the demographic attributes of age, gender and race \cite{11}. By using the appropriately trained AAM for obfuscation, the attributes for which it was trained can be preserved in the output. An alternative approach to the explicit specification of classes to be preserved involves the use of multimodal discriminant analysis to allow for the representation of identity and other attributes in orthogonal subspaces \cite{7}. This allows for $k$-anonymity to be applied within the identity subspace while preserving or modifying other attributes relevant to utility, such as age, gender and race, as desired within each of their subspaces.

In the context of preservation of facial expressions, an approach has been proposed to calculate the difference between AAM parameters of original instances having a neutral expression and the target expression (e.g., happiness) and then add this difference to an obfuscated instance with a neutral expression. In this way, the target expression can be transferred to the obfuscated image \cite{10}. Preservation of facial expressions has also been investigated for generative NNs. By designing the network architecture to allow for multiple input vectors over different types of classes, various combinations of these classes can be targeted in the output \cite{26}. This method has been applied to generate $k$-same output having specific facial expressions \cite{4}.

Recently, differential privacy has been applied to preserve privacy in images by adding noise directly to the pixel intensities \cite{5}. While the use of sufficient noise can provide a strong guarantee of privacy, this renders the output images unrecognizable as human faces. To improve upon the visual quality, an alternate approach has been explored in which singular value decomposition is employed to add noise to the singular values of a matrix of pixel intensities \cite{74}. Although visual quality is improved, the obfuscated images remain far from being photo-realistic. Furthermore, the application of noise only to the singular values, as opposed to all matrices of the decomposition, leaves a potential for information leakage which is not investigated. To the best of our knowledge, our work is the first to study differential privacy applied to generative models for the obfuscation of facial images.

\subsection{Privacy-Preserving Machine Learning: Facial Images}

For completeness, we review here some of the most relevant works on privacy-preserving machine learning, primarily in the context of differential privacy. While most of these works differ fundamentally in their goals and motivation from our own research, this review serves to clarify where our work is situated within the existing literature.

In privacy-preserving machine learning, it is necessary to prevent the model from leaking sensitive information. This has been studied in the context of various attacks such as the inference of training data (i.e., membership inference) \cite{76}, the inference of participant data in collaborative learning \cite{77,78}, and the use of trained models to draw unintended inferences or reconstruct training data (i.e., model-inversion attacks) \cite{79}. In some scenarios, differential privacy can be applied during training, e.g., by adding noise to gradients used to update model parameters \cite{80}, to produce a model that protects details about its training data. In other cases, such as with collaborative learning, it has been shown that differential privacy is not appropriate for the granularity of protection required, and an active adversary can thwart attempts at privacy protection by crafting their shared parameters to force others to reveal more sensitive information \cite{77}. While these attacks and applications of differential privacy can all be related to the protection of sensitive information in facial images, this only holds in the context of training data for machine learning models and relies on analysis related to the training and usage of said models. Contrary to this, we focus on the task of protecting facial images in any context, abstracted from their intended usage.

Differential privacy has also been applied to protect sensitive data used to train generative adversarial networks \cite{61}. When trained in a privacy-preserving manner, the network is able to learn about the distribution of the training data and subsequently generate new instances of data in such a way that the generated instances do not enable accurate inferences about the specific training instances. Variations on this concept have been applied to generate synthetic datasets of facial images in a differentially private manner \cite{62,63}. Although such approaches allow for the generation of facial images in a privacy-preserving manner, this too falls outside of the domain of facial obfuscation. The images produced from a generative network are drawn from a learned distribution, in this case a distribution of facial images, and are thus intended only to look as though they may have been drawn from the training set. This differs from the task of facial obfuscation in which the output image should retain as much information from the input image as possible short of revealing the aspects deemed sensitive.

A different direction explored in some works is the design of neural networks that reduce the leakage of sensitive information in the representations of images used for tasks such as classification. Here, the goal is not the protection of the training data but rather the protection of the input images provided to the model at inference time. This goal has been cast in an adversarial training process in which the network aims to maximize accuracy for the intended classification task while minimizing leakage of sensitive information as measured by entropy \cite{81}. A similar adversarial training approach has been explored in the context of minimizing identity classification accuracy in facial images while maximizing facial expression classification accuracy \cite{82} and while maximizing action detection accuracy \cite{83}. Although these approaches suit the goal of protecting a specific facial image, they are highly tailored to specific tasks and do not produce obfuscated images as part of their output. An alternate approach, Privacy-Protective-GAN \cite{68}, combines the framework of adversarial training for privacy with an auto-encoder style generative architecture to produce obfuscated facial images. While this fits for the task of facial obfuscation, the adversarial training objective provides no formal guarantee regarding the level of privacy that is achieved, nor is there any means to adjust the level of privacy enforced by the trained model.

\section{Weaknesses in Existing Facial Obfuscation} \label{sec:deficiences}

Given the importance of preserving privacy in images, a good method of obfuscation must assert a meaningful guarantee about the level of privacy it provides. Without such a guarantee, it is impossible to formally assess the effectiveness of the obfuscation. Empirical results may help to gain intuition on which approaches appear promising. However, without a formal guarantee to back up the results, it is impossible to assert that privacy will remain protected in untested scenarios against unknown attacks. For this reason, we focus our attention only on methods of obfuscation which offer a formal privacy guarantee.

The necessity of this restriction is underscored by the concept of parrot attacks \cite{1}. A parrot attack uses a neural network to classify identities using labeled instances of obfuscated images as the training set. Having learned about patterns in the obfuscation during training, the network is made much more effective at defeating the obfuscation. Despite pixelization being reasonably effective against human recognition and even naive machine recognition,  it can be completely defeated by a parrot attack. This formed a strong basis for the need of a formal privacy guarantee such as that provided by the $k$-same family.

The $k$-same approaches employ a privacy guarantee derived from $k$-anonymity \cite{31} which asserts that the original identity for any obfuscated image is indistinguishable from at least $k-1$ others. This guarantee is a result of the obfuscation process which draws upon clusters of $k$ or more images to produce averaged instances as replacements for all images in each cluster. This makes it impossible for any software to achieve a better probability of re-identification than $\frac{1}{k}$.

However, the $k$-same guarantee relies on assumptions about the nature of the attack. In this section, we discuss these assumptions. We show why they are often unrealistic in practice, making the guarantee weaker than it appears to be.

\subsection{Background Knowledge}

A well-known deficiency of $k$-anonymity is its susceptibility to attacks that employ background knowledge \cite{46}. This refers to cases where the attacker uses prior knowledge about the sensitive information to draw inferences that violate the privacy guarantee. For example, in the context of a database of hospital records, $k$-anonymity would typically be applied to create groups of database records that are indistinguishable on their demographic attributes but with the original medical condition (i.e., the sensitive information) of each record preserved for statistical analysis. If an attacker already knows the medical conditions of one or more individuals (e.g., friends or family members), they can eliminate the corresponding records from the anonymized groups by finding the matching demographic information and medical condition. If the removal of records reduces the size of an anonymized group to less than $k$, this violates the privacy guarantee.

This concept carries directly over to the $k$-same privacy guarantee. If, via prior knowledge, the attacker knows with certainty that some of the $k$ individuals could not be in the obfuscated image, they can discount them from the set of $k$ identities. An attacker could come by this knowledge in a number of ways: personal knowledge about friends and family, information scrapped from other data sources such as social media, etc. The simple combination of knowledge about the time at which an photo was taken and the approximate locations of some of the $k$ individuals at that time can be enough to derive a proper subset of the $k$ individuals which violates the privacy guarantee.

Contextual information in an image can often enable these types of inferences. Using signs, architecture or landscapes in an image, an attacker might recognize the location or employ software to determine it. Knowledge about locations that individuals frequent may greatly increase the probability of some possibilities over others. Similarly, if some of the $k$ identities are known to live in different cities than where the photo was taken or worse yet, different continents, these identities become much less probable. Other cues such as accessories or clothing on obfuscated individuals may also greatly impact the probabilities accorded to the $k$ possible identities. Since the privacy guarantee asserts that each of the $k$ identities are equally probable, this is also in violation of the guarantee.

We note that the original $k$-same paper does acknowledge this vulnerability to contextual information and asserts that the privacy guarantee applies strictly to the information contained within the face, not to the image as a whole \cite{1}. While this important distinction allows for the privacy guarantee to be upheld, it is a major restriction on the practical applicability of the $k$-same guarantee. Most contexts in which facial obfuscation is applied will be rich with contextual information, making the privacy guarantee much less meaningful.

\subsection{Composition Attacks}

Another deficiency of $k$-anonymity is a susceptibility to composition attacks \cite{36}. This is a class of attacks that exploit information from multiple, potentially uncoordinated, obfuscated releases to violate the privacy guarantee. A simple instance of this is the intersection attack. An attacker first identifies the clusters in which a particular individual exists from two different releases. If the releases were uncoordinated, the clusters likely differ, allowing the attacker to take their intersection to achieve a new set with a cardinality less than $k$.

This attack again carries directly over to the $k$-same approach. Consider a scenario where an individual takes a photo that they wish to upload to social media. Privacy protection might be applied to the individual or perhaps to bystanders who were captured in the background of the photo. Should the individual decide to upload the same photo to two or more social media platforms, the issue of uncoordinated obfuscation immediately arises. An attacker needs only scrape these platforms for similar photos to apply an intersection attack.

Intersection attacks may even be effective for multiple releases from the same organization if care is not taken. For example, an individual may take consecutive photos and then upload all of them. Algorithms for $k$-same determine clusters based on the similarity of faces but many factors beyond facial identity (e.g., pose, angle and lighting) could impact similarity. It is therefore not unlikely that multiple images of the same individual will result in different clusters. Sequences of images uploaded in this way would be an ideal target for intersection attacks.

Most $k$-same approaches require each individual to appear only once in the gallery of images to be obfuscated. This prevents intersection attacks for releases from the same organization but does not protect against uncoordinated releases across multiple organizations. Furthermore, enforcing this restriction may be very challenging in practice. While the primary subject in a photo might be determined based on the account used to upload the photo, other individuals in the photo cannot be correctly identified 100\% of the time. Face recognition software has not yet reached this level of accuracy. Without manual labeling, such a policy cannot be enforced. Beyond this, the restriction of one image per identity is very severe and does not match typical use cases for image sharing.

\subsection{Other Difficulties}

We discuss here two other difficulties that arise when using $k$-same obfuscation in practice. Although these difficulties do not violate the privacy guarantee, they hinder meaningful applications for $k$-same obfuscation in some contexts.

The first problem arises from the requirement of an input gallery of images. This may be appropriate for scenarios where batches of images are obfuscated but it is awkward to apply to cases where images are sporadically uploaded (e.g., in social media platforms). One might consider the use of a preloaded static gallery or even a dynamic gallery that gets updated as new images are uploaded. This, however, is not a good solution since identities can then participate in more than one cluster. Furthermore, if an attacker records information about identities known to be in the gallery, those identities can be discounted when an image is uploaded for a new identity. An alternative solution could rely on buffering uploaded images to form a gallery that can eventually be used to release a batch of obfuscated images. However, this necessitates a trade-off between the size of the gallery (and thus the quality of the output) and ability to deliver a timely service. In an era where users expect images to be uploaded instantly, this is not likely to be a manageable trade-off. The release of multiple batches also increases the chances of enabling composition attacks.

The second problem relates to the preservation of utility in the obfuscated output. Approaches that partition the gallery according to classes to be preserved (e.g., combinations of age and gender) place an even greater strain on the input gallery requirement. Working separately with the subset of images from each class greatly reduces the number of images available for clustering. Such an approach is not scalable for large numbers of classes that would be needed for finely grained attention to utility. In the worst case, some classes may be outliers in the overall distribution and could lack sufficient images to form a cluster. These classes would have to be merged with others in order to achieve the $k$-same guarantee, thus failing to achieve the desired granularity of classes.

\section{Differential Privacy for Generative Models} \label{sec:differential_privacy}

Due to the deficiencies of the $k$-same privacy guarantee in practical applications of facial obfuscation, we argue that a more robust privacy guarantee is required. Following the advances in the field of database privacy, we consider the potential of differential privacy to provide a stronger privacy guarantee. Differential privacy has real world applications in organizations including the US Bureau of Census, Google and Apple \cite{71} and has been studied as a means to comply to legal definitions of privacy such as FERPA \cite{72} and GDPR \cite{73}. In this section, we first review basic theory of differential privacy. We then adapt the privacy guarantee to fit the context of generative machine learning models for images and we formalize a framework to apply differential privacy to facial images. We discuss how the derived privacy guarantee addresses the issues identified with the $k$-same approach. Finally, we apply our framework to implement differentially private facial obfuscation using a generative NN.

\subsection{Differential Privacy for Databases}

A privacy guarantee that offers an absolute bound on re-identification risk necessitates restrictive assumptions about the attacker. This is due to the fact that it is impossible to prevent an attacker from learning about the sensitive information through means other than the obfuscated release \cite{22}. Differential privacy recognizes this difficulty and instead adopts a privacy guarantee that limits the increase in an attacker's knowledge about the sensitive information. In the context of databases, the goal is to release aggregate information about the database while preventing that information from being exploited to derive sensitive details about the individual records. Differential privacy functions by using a \emph{randomization mechanism} to add controlled noise to database query responses in order to release useful responses while achieving a desired level of indistinguishability between potential configurations of the database contents.

Two databases are considered to be \emph{adjacent} if they differ by a single record. Informally, the privacy guarantee enforces that any pair of adjacent databases must be bounded within a multiplicative factor of $e^{\epsilon}$ (where $\epsilon$ is the \emph{privacy parameter}) in their probabilities of producing the same noisy query response. This is often interpreted as a ratio of $e^{\epsilon}$ between these probabilities. With a sufficiently small ratio, similar databases have similar probability distributions over their noisy query responses, causing them to behave similarly with respect to the noisy query responses they produce. This limits the usefulness of the noisy responses as a means to distinguish between potential configurations of the database. The privacy guarantee \cite{30} in Formula \ref{eq:dp_guar} formally states this requirement in terms of any pair of adjacent databases $D_1, D_2 \in \mathbb{D}$, where $\mathbb{D}$ is the set of valid database configurations, and a randomization mechanism $K: \mathbb{D} \rightarrow \mathbb{R}^n$, where $n \in \mathbb{Z}^+$.

\begin{equation} \label{eq:dp_guar}
	\Pr \left( K \left( D_1 \right) = R \right) \leq 
	e^{\epsilon} \Pr \left( K \left( D_2 \right) = R \right) 
	\hspace{2em} \forall R \in \mathbb{R}^n.
\end{equation}

To achieve this privacy guarantee, the mechanism $K$ must take into account the value of $\epsilon$ and the \emph{query sensitivity}. The sensitivity $\Delta F$ of a query $f: \mathbb{D} \rightarrow \mathbb{R}^n$ is defined as the maximum possible $L_1$ distance between the query responses for any pair of adjacent databases. The guarantee can be achieved by adding to the query response a vector of $n$ continuous random variables, each drawn independently from a Laplace distribution with $\frac{\Delta F}{\epsilon}$ as its scaling parameter \cite{30}. The exponential decay of probability density in the Laplace distribution benefits the utility of the mechanism by limiting the expected perturbation of the query responses. Through the selection of an appropriate value for $\epsilon$, a data custodian can control how much information is revealed about the contents of the database.

\subsection{Framework for Generative Models} \label{sec:framework}

We now consider how differential privacy can be applied to generative models for images. A generative model can represent images of instances from specific classes (e.g., human faces) using a numeric representation that abstracts from pixel intensities. Our goal is to protect the privacy of individuals in images by modifying these numeric representations to prevent facial identification while maintaining utility and visual quality. Differential privacy is ideal for this purpose as it provides a robust guarantee against the accuracy of the inferences an attacker can make about the original data. The application of noise to the numeric representation of the model allows for the generation of photo-realistic instances of novel human faces. This avoids the significant degradation in visual quality that results from the addition of noise to pixel intensities.

When moving from the domain of databases to that of generative model representations, the concepts of adjacency and query sensitivity can no longer be applied for the configuration of a mechanism. In place of a database where each record is an individual, we have a numeric representation of a single individual (e.g., features extracted by the model). To protect sensitive data in this form, one can apply a generalization of differential privacy to arbitrary \emph{secrets} \cite{23}, where a secret is any numeric representation of data. In our case, the secret is the generative model representation of an individual. This generalization substitutes the notion of adjacency between databases with distance between secrets. By controlling noise according to an appropriate distance metric, the privacy guarantee is adapted to ensure that similar secrets are highly indistinguishable while very different secrets remain distinguishable. For a pair of databases, the distance between them is the number of records by which they differ. For other types of secrets, the distance metric must be carefully chosen in order to provide an appropriate privacy guarantee. An example of a well-studied instantiation of this generalization is geo-indistinguishability \cite{53} which protects users' geographic locations, represented as two-dimensional coordinates, using $L_2$ distance as the metric.

The notion of distance between secrets is appropriate for the representation of images within a generative model. Any model that employs a numeric representation of images allows for the calculation of distance between images. While the exact representation of an image differs from model to model, they can generally mapped to a vector of fixed length with little difficulty. We provide details on how this concept can be applied to generative NNs in Section \ref{sec:implementation}. To develop a general framework here, we consider the representation of an image to be a vector $X \in \mathbb{R}^n$ and the randomization mechanism to be a function $K: \mathbb{R}^n \rightarrow \mathbb{R}^n$ used to produce an obfuscated instance of the image. Although the differential privacy generalization only deals explicitly with one and two-dimensional secrets \cite{23}, its generalization to an $n$-dimensional vector is straightforward. We therefore adapt the privacy guarantee to suit this purpose in Formula \ref{eq:dist_guar}, using a distance function $d: \mathbb{R}^n \times \mathbb{R}^n \rightarrow \mathbb{R}$.

\begin{equation} \label{eq:dist_guar}
	\begin{split}
		\Pr \left( K \left( X_1 \right) = R \right) \leq 
		e^{\epsilon d \left( X_1, X_2 \right)} \Pr \left( K \left( X_2 \right) = R \right) \hspace{2em} \forall X_1, X_2, R \in \mathbb{R}^n.
	\end{split}
\end{equation}

Comparing this to Formula \eqref{eq:dp_guar}, the databases $D_1$ and $D_2$ have been replaced by secrets $X_1$ and $X_2$ and the distance function now appears in the exponent of the multiplicative factor $e^{\epsilon}$. The distance between any pair of secrets acts as a coefficient to $\epsilon$ when interpreting the ratio of their probabilities. Intuitively, the meaning is that the more similar a pair of images are to each other, the harder is it to determine which of them led to a given obfuscated instance. This hampers the accuracy with which attempts at re-identification can be made. To achieve this guarantee, we must first determine an appropriate distance metric to measure the distinguishability of the numeric representations of images.

A natural choice for the distance metric is $L_1$ distance, however, we must be wary of the meaning of each element in the vectors. Should certain elements have differently sized ranges, they should be obfuscated using different magnitudes of noise. If one element has a much larger range than the others, the addition of noise configured to the smaller range would do little to prevent an inference of high accuracy on the original value of the element. We therefore apply normalization such that the distance between any pair of elements in the $i^{th}$ position of a pair of vectors falls within the range $[0,1]$. Letting $R_i = [i_{min}, i_{max}]$ be the range of elements in the $i^{th}$ position of a model representation vector, we define a normalized, element-wise distance metric as follows:

\begin{equation} \label{eq:ele_dist}
	d_i \left( x, x' \right) = 
	\frac{\left| x - x' \right|}{i_{max} - i_{min}}
	\hspace{2em} \forall x, x' \in R_i.
\end{equation}

A distance metric for vectors defined as the sum of the element-wise distances for each position would be appropriate for images represented by the same model. However, a more useful framework would allow for reasoning about the level of privacy across different models. Ideally, the meaning of a privacy parameter $\epsilon$ applied to one model should have a similar meaning for a different model. For this, we require another normalization to account for models having vectors of different lengths. We therefore define the distance metric for vectors as follows:

\begin{equation} \label{eq:vec_dist}
	d \left( X_1, X_2 \right) = 
	\frac{\sum\limits_{i=1}^n d_i \left( X_{1i}, X_{2i} \right)}{n}
	\hspace{2em} \forall X_1, X_2 \in \mathbb{R}^n.
\end{equation}

By using this distance metric in combination with Formula \ref{eq:dist_guar}, we obtain a meaningful privacy guarantee for the model representations of images. Although this type of metric is, in general, not novel, its use in this context is. We must therefore address how to configure a mechanism to satisfy this instantiation of the privacy guarantee. This leads to our main result in the development of a framework for the application of differential privacy to generative models for images.

\begin{theorem} \label{th:lap}
	Any image $X \in \mathbb{R}^n$ can be protected by $\epsilon$-differential privacy through the addition of a vector $\left( Y_1, ..., Y_n \right) \in \mathbb{R}^n$ where each $Y_i$ is a random variable independently drawn from a Laplace distribution using a scaling parameter $\sigma_i = \frac{n \left( i_{max} - i_{min} \right)}{\epsilon}$.
\end{theorem}

\begin{proof}
	
	We must satisfy the privacy guarantee (Formula \ref{eq:dist_guar}) using our proposed distance metric (Formula \ref{eq:vec_dist}). The form this privacy guarantee takes is our starting point in Formula \ref{eq:proof_2}. Through manipulation of this inequality and the substitution of mechanism probabilities with a Laplace distribution, we prove that the selection of an appropriate scaling parameter for each instance of the Laplace distribution allows for the privacy guarantee to be satisfied.
	
	{\small
		
		\begin{equation} \label{eq:proof_2}
			\begin{split}
				\prod\limits_{i=1}^n \Pr \left( K \left( X_{1i} \right) = R_i \right) \leq 
				e^{\frac{\epsilon \sum\limits_{i=1}^n d_i \left( X_{1i}, X_{2i} \right)}{n}}
				\prod\limits_{i=1}^n \Pr \left( K \left( X_{2i} \right) = R_i \right) \hspace{2em} \forall X_1, X_2, R \in \mathbb{R}^n.
			\end{split}
		\end{equation}
		
		\begin{equation} \label{eq:proof_3}
			\begin{split}
				\prod\limits_{i=1}^n \Pr \left( K \left( X_{1i} \right) = R_i \right) \leq 
				\prod\limits_{i=1}^n e^{\frac{\epsilon d_i \left( X_{1i}, X_{2i} \right)}{n}}
				\prod\limits_{i=1}^n \Pr \left( K \left( X_{2i} \right) = R_i \right) \hspace{2em} \forall X_1, X_2, R \in \mathbb{R}^n.
			\end{split}
		\end{equation}
		
		\begin{equation} \label{eq:proof_4}
			\begin{split}
				\prod\limits_{i=1}^n \frac{e^{-\frac{\left| X_{1i}, R_i \right|}{\sigma}}}{2 \sigma} \leq 
				\prod\limits_{i=1}^n e^{\frac{\epsilon d_i \left( X_{1i}, X_{2i} \right)}{n}}
				\prod\limits_{i=1}^n \frac{e^{-\frac{\left| X_{2i}, R_i \right|}{\sigma}}}{2 \sigma} \hspace{2em} \forall X_1, X_2, R \in \mathbb{R}^n.
			\end{split}
		\end{equation}
		
		\begin{equation} \label{eq:proof_6}
			\begin{split}
				\prod\limits_{i=1}^n e^{\frac{ \left| X_{2i}, R_i \right| - \left| X_{1i}, R_i \right|}{\sigma}} \leq 
				\prod\limits_{i=1}^n e^{\frac{ \left| X_{2i} - X_{1i} \right|}{\sigma}} \leq
				\prod\limits_{i=1}^n e^{\frac{\epsilon d_i \left( X_{1i}, X_{2i} \right)}{n}} \hspace{2em} \forall X_1, X_2, R \in \mathbb{R}^n.
			\end{split}
		\end{equation}
		
		\begin{equation} \label{eq:proof_7}
			\begin{split}
				\prod\limits_{i=1}^n e^{\frac{\epsilon d_i \left( X_{1i}, X_{2i} \right)}{n}} =
				\prod\limits_{i=1}^n e^{\frac{ \epsilon \left| X_{2i} - X_{1i} \right|}{n \left( i_{max} -i_{min} \right) }}
				\hspace{2em} \forall X_1, X_2.
			\end{split}
		\end{equation}
		
	}
	
	From Formula \ref{eq:proof_7}, it becomes clear that the inequality holds when using an independent Laplace distribution for each pair of elements $X_{1i}, X_{2i},$ substituting the scaling parameter $\sigma$ with a corresponding value $\sigma_i = \frac{n \left( i_{max} - i_{min} \right)}{\epsilon}$.
	
\end{proof}

Using the generalization of differential privacy, the notion of query sensitivity is implicitly captured in the distance metric. Since the distance metric of Formula \eqref{eq:vec_dist} has a range of $\left[ 0, 1 \right]$, the ratio of probabilities for a pair of maximally dissimilar images to produce the same obfuscated output is $e^{\epsilon}$. This is akin to the meaning of the privacy guarantee for a pair of databases that differ on every record. In order to select an appropriate value of $\epsilon$, a data custodian must keep in mind that similar images will have a very small distance between them, requiring much larger values of $\epsilon$ to provide a reasonable ratio. In Section \ref{sec:experiments}, we demonstrate the implications of the choice of $\epsilon$ on the levels of privacy and utility.

\subsection{Benefits of Differentially Private Facial Obfuscation}

We now describe the improvements we obtain from the use of differential privacy for each of the problems identified in Section \ref{sec:deficiences}.

\subsubsection{Background Knowledge}

By removing dependence of the attack model on an absolute level of re-identification risk, we are able to reason about the level of privacy in the presence of attackers with background knowledge. If the location in a photo is identified as a particular city, no facial obfuscation can prevent the inference that individuals living in the identified city have a higher probability of being the obfuscated identity than individuals living elsewhere. Yet, the differential privacy guarantee continues to hold as the background knowledge does not impact the conditional probability distribution used by the randomization mechanism. Since the privacy guarantee concerns only the change in the attacker's knowledge when presented with the obfuscated data (e.g., the face), it is unaffected by other sources of information the attacker may gain access to.

\subsubsection{Composition Attacks}

Another very important property of differential privacy is its resilience to composition attacks. The composition theorem \cite{30} states that for two differentially private releases using privacy parameters $\epsilon_1$ and $\epsilon_2$ respectively, the privacy guarantee holds for a privacy parameter $\epsilon = \epsilon_1 + \epsilon_2$. Thus, even in the case of uncoordinated releases, we still have a valid privacy guarantee. Furthermore, this removes the restriction on the same individual appearing only once in the release of obfuscated images.

\subsubsection{Input Image Gallery}
Differentially private image obfuscation has no need for a gallery of images in order to perform obfuscation. Since noise is added on a per-image basis, there is no computation of clusters required. Given a trained model, obfuscation of a single image or a batch of images can be performed with ease. This makes the obfuscation process much more versatile.

\subsection{Privacy Guarantee Interpretation} \label{sec:priv}

Although the interpretation of the differential privacy guarantee is relatively well understood in the context of databases, the distance-generalized guarantee which we employ changes the interpretation of the privacy parameter $\epsilon$. To assist users and data curators in understanding the implications of a chosen privacy parameter, we provide a brief discussion here on the generalized privacy guarantee.

Recall that the generalized privacy guarantee replaces databases with arbitrary secrets and scales $\epsilon$ by the distance between any pair of secrets for which the guarantee is to be interpreted. The distance between any pair of databases can in fact be interpreted as a Hamming distance (i.e., the number of records by which a pair of databases differ), allowing for the generalization to capture the standard interpretation of differential privacy. At distance 1, this corresponds to the basic privacy guarantee and at distance $d > 1$, this corresponds to $d$ transitive applications of the privacy guarantee, resulting in a multiplicative bound of $e^{d \epsilon}$. For databases with $n$ records, the range of the distance metric is $[1,n]$.

In contrast, the distance metric that we employ is bounded by the range $[0,1]$. Although there is no direct correspondence between distance measured for databases and distance measured for other domains, it is clear that interpretation of the privacy guarantee for a pair of maximally distant secrets (e.g., a pair of databases which differ on every record) has little meaning in practice. Just as differential privacy is interpreted in terms of similar pairs of databases, it should also be interpreted in terms of pairs of images encodings that are that are similar to each other. In the context of facial images, the intuition is that the obfuscated encoding should be difficult to distinguish for other similar encodings, rendering it difficult to associate with a specific individual. We make no claim about an exact distance between secrets for which the guarantee should be interpreted but propose that users should examine the implications of the privacy guarantee for small distances to gain intuition on the level of distinguishability that will be achieved between secrets at various distances from each other. For example, for pairs of secrets at distance 0.01 from each other, the privacy guarantee enforces a multiplicative bound of $e^{0.01 \epsilon}$. While this may seem like a rather small distance, this is in fact a much larger fraction of the total range than is often used in the context of databases where there may be thousands or even hundreds of thousands of records.

\subsection{Implementation Details} \label{sec:implementation}

Generative NNs have the useful property of producing photo-realistic images. We now describe how our framework can be applied to these models. Provided that the addition of noise is properly controlled, the output will be a photo-realistic image of any newly created identity.

\begin{figure*}
	\centering
	\includegraphics[width=1\textwidth]{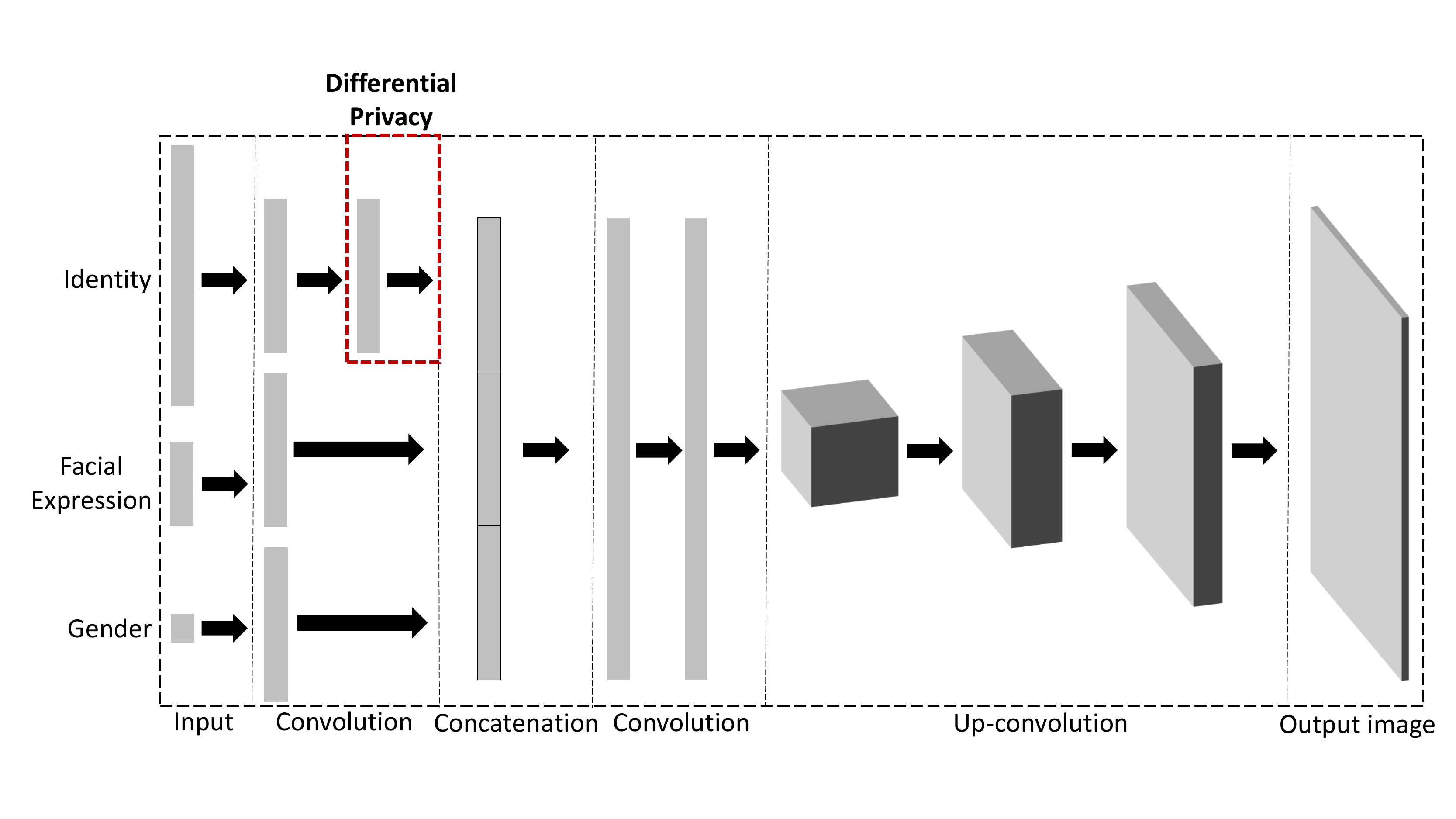}
	\caption{Visualization of the layer architecture in an up-convolutional neural network using differential privacy. Noise is applied to the output of the second identity layer. The numbers and shapes of the convolutional layers shown here are not exact and represent only the general structure of such a network.}
	\label{fig:architecture}
\end{figure*}

We consider network architectures that take one or more class vectors as input and employ up-convolution to transform the input into a visual representation in pixel space \cite{26}. By considering each identity to be a different class, an input vector can specify the individual to be generated. The identity class vector is an obvious choice as the model vector to be obfuscated. However, this leads to some form of interpolation between the identities. To apply a finer degree of modification to the identity, we propose the application of obfuscation at the second layer of the network. Typically, the second layer applies convolution to the class vector and transforms it into a vector of high-level numeric facial features. By applying obfuscation to these features instead, we can achieve a richer variety in the potential modifications to the face. We therefore apply obfuscation to the output of the first convolutional layer of the network and pass the obfuscated feature vector on as the input to the next layer of the network. A sample architecture is shown in Figure \ref{fig:architecture}.

Our proposed implementation can be interpreted as an extra layer added to the network which is used only after training is complete when obfuscation is to be applied. The layer has no weights and is simply an application of the Laplace mechanism, configured as described in Theorem \ref{th:lap}. The noisy output, then conforming to the differential privacy guarantee, is passed on to the next layer of the network. All other propagation through the network occurs as normal. By ensuring that the encoding of the facial identity has been passed through the layer implementing the Laplace mechanism, we are able to produce an output image depicting a facial identity that has been obfuscated in a differentially private manner.

Information about the range of each model vector element can be used as a means to preserve the visual quality of the obfuscated output. Noisy elements that have gone too far beyond the valid range may lead to visual artifacts or distortions in the output image. To prevent this, we snap any out-of-bounds noisy value back to the nearest valid value. Since differential privacy is resistant to any form of post-processing \cite{30} and the ranges of the elements are non-sensitive information, this step cannot violate the privacy guarantee.

\subsection{Obfuscation of Images with Unknown Identities} \label{sec:unseen}

When designing a system for the obfuscation of identities in images, an important consideration is the ability to obfuscate any identity. However, the generative NN architecture we employ does not directly allow for the representation of classes that were not learned during the training process. Since each identity is a different class, this means that the model cannot directly represent unknown identities and thus cannot obfuscate them.

To solve this problem, we propose to formulate the approximation of an appropriate input vector for an unknown identity as an optimization problem. Since the generative NN has learned a representation of each training identity, a new identity can be approximated as a weighted sum of the known identities. When provided with these weights as the input identity vector, the generative NN would produce an interpolation between the identities which can act as an approximate visual representation of the unknown identity. In order to formalize this concept as an optimization problem, we must select an appropriate representation for the identities. This cannot be done using the generative NN feature vectors since the representation of the target identity is unknown. We therefore employ a secondary neural network that has been trained for classification of facial identity. By removing the final layer of the classification network, its output becomes a high-level vector of facial features which we can use as the representation of an identity. The significance of using a classification network for this purpose is that it need not have seen any of the target identities in its training data in order to produce feature vectors for them. Such a network is therefore ideal to provide feature vectors for the identities from the generative NN training data as well as for the target unknown identity.

We now formalize the optimization problem using real-valued feature vector representations of identities from a classification network. Let $\mathbb{X}$ be a set of n-dimensional feature vectors representing $m$ identities on which the generative NN was trained (i.e., $\left| \mathbb{X} \right| = m$) and let $Y$ be an n-dimensional feature vector representing an identity that is unknown to the generative NN. Our goal is to determine a set $W$ of weights that minimize the distance between a weighted sum of the vectors in $\mathbb{X}$ and the vector $Y$.

\begin{equation} \label{eq:opt_problem}
	\min\limits_{W} \left\| \sum\limits_{i=1}^{m} W_i \mathbb{X}_i - Y \right\|
\end{equation}

While this optimization problem is similar to the format of an objective function for a linear program, the necessity of absolute value calculations for the $L_1$ distance between the vectors prevents this from being written as a linear function. However, this problem, known as least absolute deviations, can be rewritten in an alternate but equivalent formulation that avoids the need for absolute value functions through the introduction of additional variables \cite{70}. This formulation is as follows:

\begin{equation} \label{eq:lp_obj}
	\min \sum\limits_{i=1}^{n} u_i
\end{equation}

Subject to the following constraints:

\begin{equation} \label{eq:lp_con}
	u_i \geq y_i + \sum\limits_{j=1}^{m} W_j \mathbb{X}_j \hspace{2em} i = 1, ..., n
\end{equation}

\begin{equation} \label{eq:lp_co2n}
	u_i \geq - \left( y_i + \sum\limits_{j=1}^{m} W_j \mathbb{X}_j \right) \hspace{2em} i = 1, ..., n
\end{equation}

The added constraints ensure that the values assigned to the new variables respect the absolute value functions from the original problem. This formulation can be given to any linear programming solver in order to find the optimal weights for the approximation of the unknown identity represented.

\section{Differential Privacy in Pixel-Space} \label{sec:pixel}

While we have thus far considered the application of differential privacy to the numeric representations of generative models, it can also be applied directly to the pixel intensities of an image. Such an approach suffers in the visual quality of the obfuscated images. Noise added in this way is no longer guided to obfuscate only specific aspects such as identity. Yet, by discarding the use of a generative model, the randomization mechanism can be applied to any image, regardless of what is depicted. This versatility allows for obfuscation to be applied to images that are not readily captured by available models. For example, it may be desirable to obfuscate signs, license plates or complete vehicles. A pixel-space randomization mechanism provides a means to directly achieve privacy protection on any image.

\subsection{Obfuscation via the Laplace Mechanism}

Differential privacy via a Laplace mechanism has previously been applied in combination with pixelization to achieve pixel-space obfuscation \cite{5}. In this work, the authors defined two images to be adjacent if they differed by $n$ pixels, where the value of $n$ is chosen by the user based on the size (in pixels) of a window that covers the portion of the image deemed to be sensitive information. For example, if a face covers roughly a $100 \times 100$ pixel space and the user wishes to obfuscate the face, $n$ would be set to 10000. Following the standard framework for differential privacy, the query sensitivity is then defined as the maximum possible difference between any pair of adjacent images. For an image with $c$ channels, each of which covers a range of $k$ possible pixel intensities, the query sensitivity is $(k-1)nc$. Given this query sensitivity, $\epsilon$-differential privacy can be provided through an independent application of the Laplace mechanism to each pixel in each channel using a scaling parameter of $\frac{(k-1)nc}{\epsilon}$ \cite{5}. To reduce the query sensitivity, the authors propose to apply pixelization prior to obfuscation. For a specified level of pixelization $b \in \mathbb{Z} \hspace{0.1cm} | \hspace{0.1cm} 1 \leq b$, the image is divided into $b \times b$ grids of pixels where each grid is set to the average of its pixel intensities. By treating each grid as a single uniform value to be obfuscated, the required scaling parameter for the Laplace mechanism is reduced to $\frac{(k-1)nc}{\epsilon b^2}$.

\subsection{Obfuscation via the Exponential Mechanism}

While the Laplace mechanism is configured based on differences in pixel intensities, we propose to instead apply the exponential mechanism \cite{55} in order to control the distribution of the obfuscated output using a measure of visual quality for potential outputs of the mechanism. Given an input $I$ (in our case an original image), the probability distribution over the set of potential obfuscated outputs $\mathbb{I}$ is defined as:

\begin{equation} \label{eq:exponential}
	\Pr \left( I' \right) \propto
	e^{\epsilon q \left( I, I' \right)}
	\hspace{2em} \forall I' \in \mathbb{I},
\end{equation}

where $q$ is a quality function that measures the utility of $I'$ with respect to $I$. In other words, $q$ is a measure that reflects the usefulness of an obfuscated output $I'$ given its original value $I$. Due to this, the resulting distribution over the potential obfuscated outputs explicitly gives preference to outputs with higher utility as measured by $q$.

We propose the use of the structural similarity index measure (SSIM) \cite{54} as the utility function used to control the mechanism. This provides a human-centric measure of image quality which is designed to match up with how the human visual system perceives information. SSIM incorporates aspects of luminance, contrast and structure into a measure that is averaged over a sliding window intended to mimic how human eyes scan a large area but focus only on local areas. The use of SSIM allows us to produce a mechanism which achieves differential privacy while explicitly favouring obfuscated results with higher visual quality, as perceived by the human visual system. For instance, while a contrast-shifted image may yield a relatively high change in pixel intensities, the visual change in the image is far less significant than an equivalent change in pixel intensities distributed at random across the image. The Laplace mechanism obfuscates images using a probability distribution based on differences in pixel intensities which fails to capture this concept of visual quality. Contrarily, the exponential mechanism using SSIM as its quality function directly reflects this notion of visual quality in its distribution over its obfuscated output. A human-centric notion of utility is of particular importance for obfuscation employed in contexts such as social media and image sharing platforms where user experience is a key concern.

\subsection{Handling the Complexity Exponential Mechanism}

A difficulty arises with the exponential mechanism in its implementation since the creation of the distribution used by the mechanism often requires explicit calculation of the quality score associated with each potential obfuscated output. Given a image containing $n$ pixels over $c$ channels, each of which covers $k$ possible pixel intensities, there are a total of $k^{nc}$ possible states an image could take on. Enumeration of these states is computationally intractable, even for very small images. For example, a $10 \times 10$ pixel RBG image has $256^{300}$ possible states. Therefore, to implement an SSIM-based exponential mechanism, we must reduce the number of states to a tractable size.

Since SSIM is intended to simulate the way the human visual system focuses only on small areas at any one time, it is calculated for a small window of pixels, typically $11 \times 11$. The overall measure is calculated by sliding this window over the image in single pixel steps and averaging the values measured at each step. Similar to this process, we propose to apply the exponential mechanism to a $p \times p$ pixel window which moves over the image in steps of $p$ pixels in order to ensure that there are no overlapping applications of obfuscation. We then coarsen the granularity of the intensities that the pixels are allowed to take on to $k' < k$ possible values, where $k$ is the original number of possible intensities. Since differentially private obfuscation reduces the accuracy of the sensitive data by design, it is often unnecessary to preserve a high degree of precision. The reduction in the granularity of pixel intensities therefore has little impact on the quality of the obfuscated output.

Under these specifications, there are $k'^{p^2}$ possible states that an application of the mechanism must consider. We have determined experimentally that $p=3$ and $k'=4$ acts as a reasonable configuration. Further increase of either value quickly renders the mechanism too computationally expensive for a standard desktop computer while further reduction of either value gives a poor approximation of SSIM. Although the use of 4 possible pixel intensities may seem very restrictive, it is important to note that the mechanism is applied independently to each channel of the image, just as SSIM is averaged over each channel. As a result, for an RGB image, each pixel can assume $4^3=64$ possible colour states. Given the large amount of noise typically required to obfuscate the images, we find this to be an acceptable level of precision.

\subsection{Privacy Configuration of the Exponential Mechanism}

We now explain how to configure and utilize the exponential mechanism in order to achieve $\epsilon$-differential privacy for image obfuscation.

\begin{theorem} \label{th:exp}
	Given an input image with $c$ channels and $n$ pixels in each channel, our proposed implementation of the exponential mechanism is able to produce an obfuscated image satisfying $\epsilon$-differential privacy through the application of the mechanism to each non-overlapping $p \times p$ pixel grid in each channel of the image using a privacy parameter of $\epsilon'=\frac{\epsilon p^2}{2nc}$ for each such application.
\end{theorem}

\begin{proof}
	
	An application of the exponential mechanism using a privacy parameter of $\epsilon'$ provides $2 \epsilon' \Delta q$-differential privacy \cite{55} where $\Delta q$ is the maximum possible change in the value of the quality function. Since the measure of SSIM takes on a value in the range $[0,1]$, the value of $\Delta q$ is 1. To obfuscate $n$ pixels over $c$ channels using an exponential grid covering $p \times p$ pixels on each application, we require $\frac{nc}{p^2}$ applications of the mechanism. By the composition theorem \cite{30}, this results in $\frac{2nc \epsilon'}{ p^2}$-differential privacy. Thus to enforce differential privacy for a user-specified privacy budget of $\epsilon$, we configure the mechanism to use $\epsilon'=\frac{\epsilon p^2}{2nc}$ at each application.
	
\end{proof}

For instances where an image has dimensions $w \times h$ such that one or both of the dimensions are not a multiple of $p$, we first apply the exponential mechanism within a $w' \times h'$ space where $w'$ and $h'$ are the largest multiples of $p$ such that $w' \leq w$ and $h' \leq h$. Then in the remaining $w-w'$ columns and $h-h'$ rows, we apply the Laplace mechanism of \cite{5}. 

As with the Laplace mechanism, our proposed application of the exponential mechanism requires a high degree of noise to achieve differential privacy. We therefore employ the pixelization trick of \cite{5} to help mitigate this. To do so, we first perform pixelization to create $b \times b$ grids of uniform pixel intensities. We then consider each such grid as a single cell of the $p \times p$ grid used by the exponential mechanism. While the mechanism still operates on $p^2$ cells, the exponential grid covers a $pb \times pb$ pixel space. In this way, using the same privacy budget $\epsilon$, we are able to increase the allocation of the budget for each application of the mechanism to $\epsilon'=\frac{\epsilon p^2 b^2}{2nc}$.

\section{Experiments} \label{sec:experiments}

In this section, we run a series of experiments to gain insight on the performance of our proposed methods of obfuscation in practice. In the context of obfuscation using generative models, we compare our proposed implementation of Section \ref{sec:implementation} to a $k$-same implementation following the design of $k$-same-net \cite{4}. We employ these comparisons to observe the relative performances of differential privacy and $k$-same obfuscation in terms of a trade-off between re-identification risk and utility. We additionally test the resilience of differential privacy against parrot attacks and composition attacks. Finally, we compare our implementations of the more general pixel-space exponential mechanism to the Laplace mechanism \cite{5}.

\subsection{Generative Model Training and Datasets}

For our generative NN implementation, we have built on top of the DeconvFaces \cite{25} network which implements the concept of up-convolution for the generation of images of input classes \cite{26}. We apply differential privacy as described in Sections \ref{sec:framework} and \ref{sec:implementation}. For the $k$-same obfuscation, we use the same generative NN implementation and, follow the approach of $k$-same-net \cite{4}, using clustering as described for $k$-same-m \cite{2}. This deviates from the use of a proxy gallery as described for $k$-same-net. It is important to note that, while a proxy gallery can reduce re-identification risk, it involves a step that is not captured by the $k$-same privacy guarantee. Thus, in the absence of a privacy guarantee that incorporates this detail, we omit the use of a proxy gallery in order to focus our experiments on the formalized aspects of the privacy guarantees. Similarly, although we have shown in Section \ref{sec:unseen} how a generative NN can be used to apply differential privacy to unknown identities, we use only identities from the training set in our experiments. Deviations from an original identity induced by the approximation process may lead to a further decrease in re-identification risk in a manner that is not captured by the formal privacy guarantee. We emphasize that this does not imply that obfuscated images of unknown identities are not protected by the privacy guarantee; merely that imperfect approximations of unknown identities may lead to empirical results that suggest a stronger level of privacy than what is in fact guaranteed by differential privacy. We have therefore made these choices in the interest of comparing strictly the privacy protection achieved due to the formalized aspects of the methods of obfuscation.

We have additionally compared differential privacy against $k$-same obfuscation using an AAM as the generative model. However, the results showed that an AAM that modifies only a tightly-cropped portion of the face serves as a poor generative model for facial obfuscation. As such, we do not include AAMs in our current experiments. The reader is referred to \cite{84} for these experiments.

We apply each method of obfuscation to two different datasets - RAFD \cite{16} and KDEF \cite{17}. These datasets provide frontal facial images of subjects wearing same coloured shirts. The use of same coloured shirts prevents bias in re-identification from the exploitation of information in unique clothing. The RAFD and KDEF datasets contain images of 67 and 70 subjects, respectively, and provide a variety of facial expressions. Due to apparent issues with lens exposure in the KDEF dataset, we have removed two of the subjects from our experiments.

The generative NN architecture accepts class vectors for identity and facial expression as input. The RAFD and KDEF datasets are therefore highly suitable for this network. We have trained the network for 1000 epochs on each of the datasets to obtain models capable of reproducing these identities. This training process is represented by Step 1 of Figure \ref{fig:testing_process}. An example of obfuscated output is shown in Figure \ref{fig:rafd_images}.

\begin{figure}[h]
	\centering
	\includegraphics[width=0.45\textwidth]{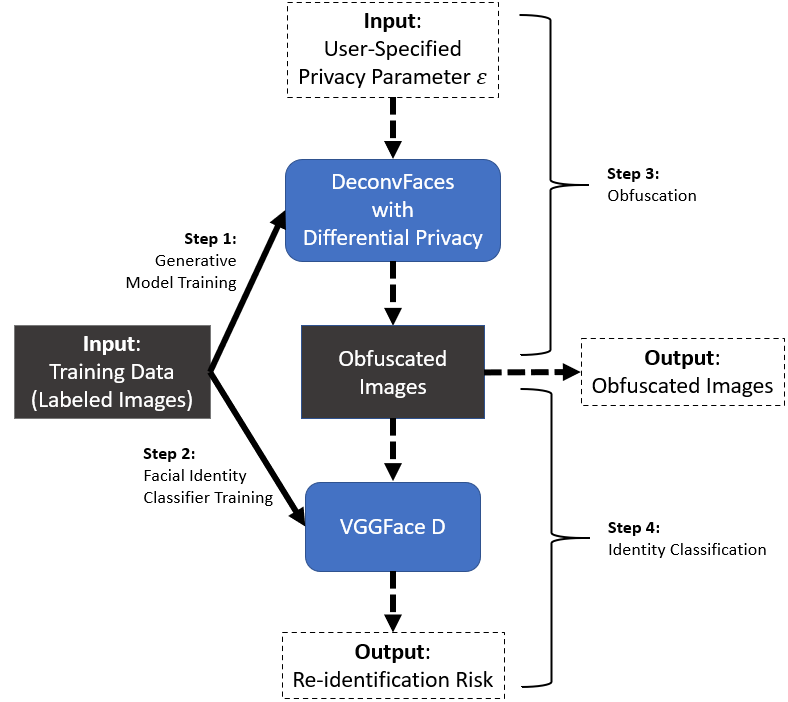}
	\caption{Illustration of the steps required for the training and testing process of the generative model and facial identity classification model used in our experiments.}
	\label{fig:testing_process}
\end{figure}

Although we do not employ the generative NN approximation of unknown identities in our experiments, we include an example of some approximations in Figure \ref{fig:unseen_approx} for reference. These approximations were produced by setting aside 10 identities from the RAFD dataset and training the generative NN on the remaining 57. We then ran a linear program for each of the 10 omitted identities to approximate them as a weighted sum of the 57 training identities. As the training set is relatively small, we expect that the quality of the approximation can be improved through the use of larger and more diverse training sets. We also note that any minor deviations in the approximated identity are not of great significance since the depicted identity is ultimately to be obfuscated.

\begin{figure*}[h]
	\centering
	\includegraphics[width=0.9\textwidth]{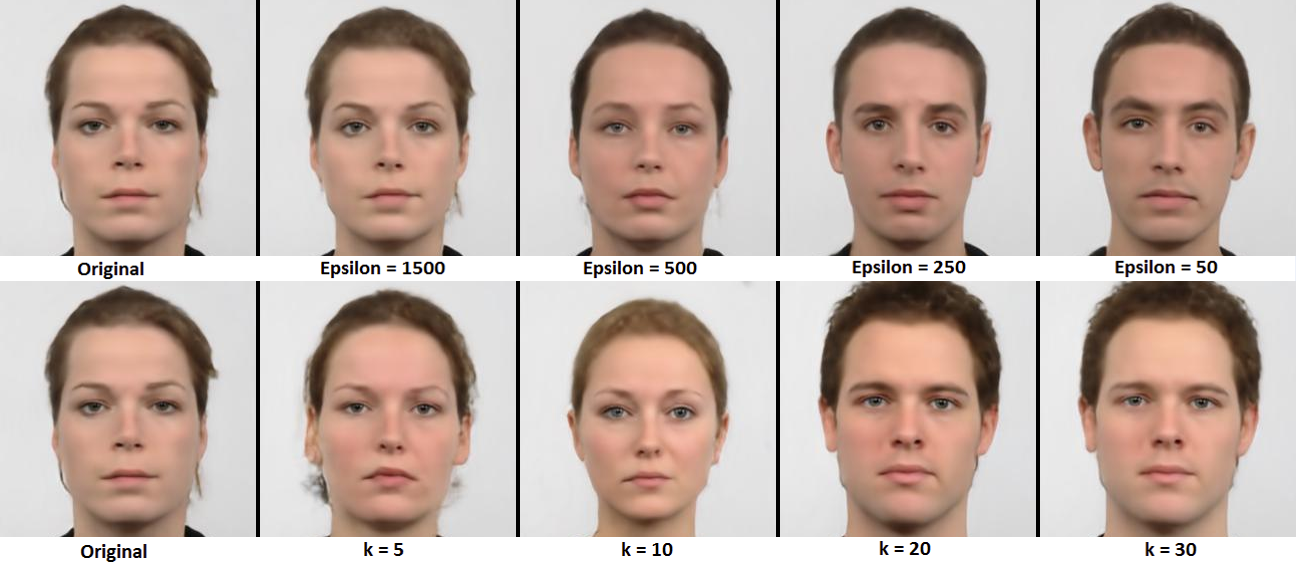}
	\caption{Obfuscation via the generative NN on the RAFD dataset. The top row employs differential privacy and the bottom row employs $k$-same obfuscation.}
	\label{fig:rafd_images}
\end{figure*}

\begin{figure*}
	\centering
	\includegraphics[width=1\textwidth]{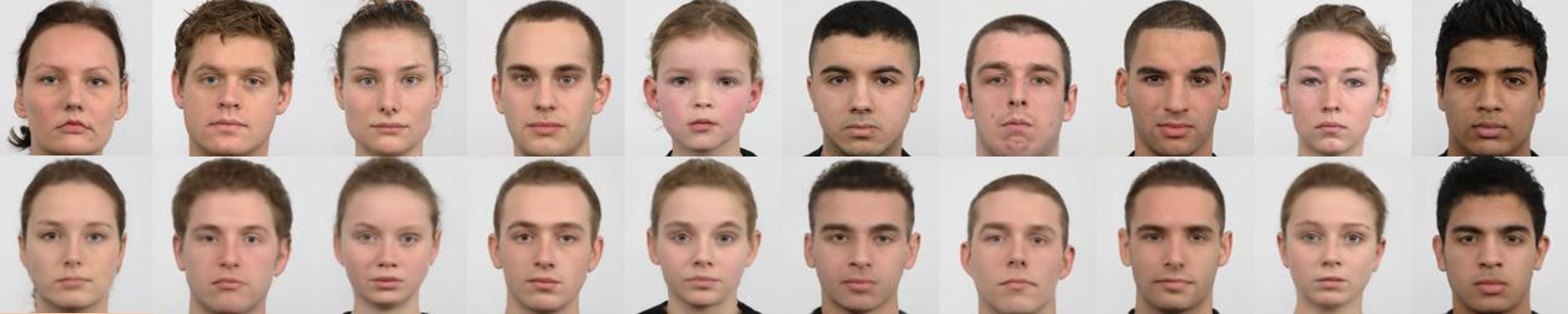}
	\caption{Approximation of unknown identities. Original images are shown in the top row with the corresponding approximations in the bottom row.}
	\label{fig:unseen_approx}
\end{figure*}

\subsection{Re-identification Risk} \label{rec:reidentification}

A good method of facial obfuscation must be able to produce obfuscated images which cannot be accurately re-identified. In other words, the risk that an obfuscated image can be associated with the identity of the originally depicted individual must be kept acceptably low. To measure re-identification risk, we have employed VGGFace D \cite{19}, a deep convolutional neural network which has been shown to achieve excellent facial identity classification accuracy. This simulates how an attacker might leverage machine learning models to launch an attack on obfuscated images. We have trained a separate model for each dataset, using the neutral and sad expressions for each identity for validation and the remaining expressions for training. Following a later released note about the network training \cite{44}, we employ Xavier initialization \cite{45} for the layer weights. To improve the robustness of the models, we have also augmented the datasets by creating two additional versions of each image - one with increased contrast and one with decreased contrast. The training of this network is represented by Step 2 of Figure \ref{fig:testing_process}. For reference, we provide the network architecture in Table \ref{tab:vgg} of Appendix A.

In our experiments, we generate obfuscated images having a neutral facial expression. We measure re\-/identification risk based on the accuracy of the top 1 guesses of the VGGFace network. Given that differential privacy is a stochastic process, for each combination of a privacy parameter and an identity to be protected, we have generated 10 obfuscated instances over which we take the average of the re-identification risk. We measure overall re-identification risk for a given privacy parameter as the average risk over all individuals in the dataset. The obfuscation and re-identification processes are respectively shown by Steps 3 and 4 in Figure \ref{fig:testing_process}. Since the $k$-same approaches are deterministic, we produce only a single output image per identity and then take the average re-identification risk over the whole dataset. We additionally measure the baseline identity classification accuracy on the original (i.e., unobfuscated) data for reference. The results are shown in Figure \ref{fig:privacy}.

\begin{figure*}
	\centering
	\includegraphics[width=1\textwidth]{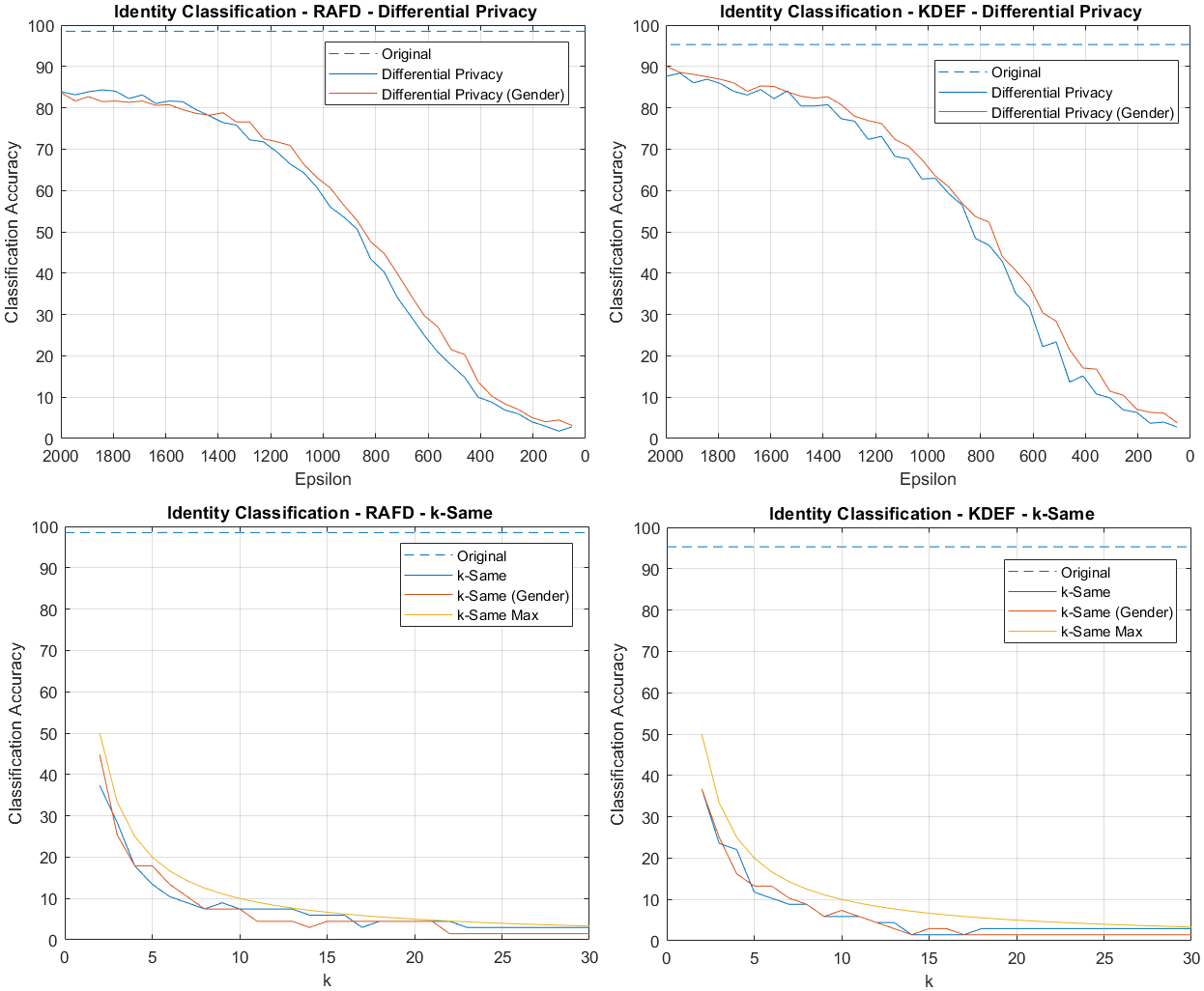}
	\caption{Identity classification accuracy for the methods of obfuscation}
	\label{fig:privacy}
\end{figure*}

Under differential privacy, a lower value of $\epsilon$ implies a stronger level of privacy whereas with $k$-same obfuscation, a higher value of $k$ implies stronger privacy. The plotted data in Figure \ref{fig:privacy} shows the expected trend of reduced re-identification risk as the level of privacy, determined by the respective privacy parameters, is strengthened. In contrast to the typical $\epsilon$ values applied to differentially private mechanisms for databases, the values used in our experiments may appear unusually high. The larger magnitude is simply a side-effect of the normalization for the model vector, resulting in the interpretation of $\epsilon$ on a different scale. Refer back to Section \ref{sec:priv} for discussion on the interpretation of $\epsilon$ in this context.

\subsection{Preservation of Utility}

To compare the methods of obfuscation in terms of utility, we focus on the ability to extract useful, non-sensitive information from the obfuscated output. Specifically, we measure classification accuracy for gender and facial expression in the obfuscated images. A favourable trade-off between privacy and utility occurs when facial identity is protected while simultaneously preserving other useful information. Thus, in our experimental setting, it is desirable to achieve low re-identification risk with high classification accuracy for the selected attributes of gender and facial expression.

We begin with the popular task of gender recognition \cite{67}. As forms of demographic classification may be desirable for data mining purposes, we consider high classification accuracy to reflect good utility. To this end, we employ a pre-trained model of convolutional neural network for the classification of gender in facial images \cite{21}. Our intent is to compare differential privacy and $k$-same obfuscation in terms of a privacy-utility trade-off. Yet, the two methods of obfuscation use proprietary privacy parameters which cannot be directly compared as a measure of privacy. We therefore plot the gender classification accuracy as a function of identity classification error in order to abstract away from the proprietary privacy parameters.

To highlight the ability of generative NNs to incorporate properties relevant to image utility into the network architecture, we have also created a modified version of the architecture that preserves gender in the obfuscated output. To do so, we have created an input layer having two classes that specify the gender in the image. By training a model with gender labels, it learns to separate features relevant to gender from those relevant to identity. This enables us to focus obfuscation only on the features relevant to identity while leaving the gender feature vector untouched. An example of gender-preserving obfuscation is shown in Figure \ref{fig:rafd_gen_images}. In the interest of a fair comparison between the methods of obfuscation, we employ the modified architecture both for differential privacy as well as $k$-same obfuscation.

\begin{figure*}
	\centering
	\includegraphics[width=0.9\textwidth]{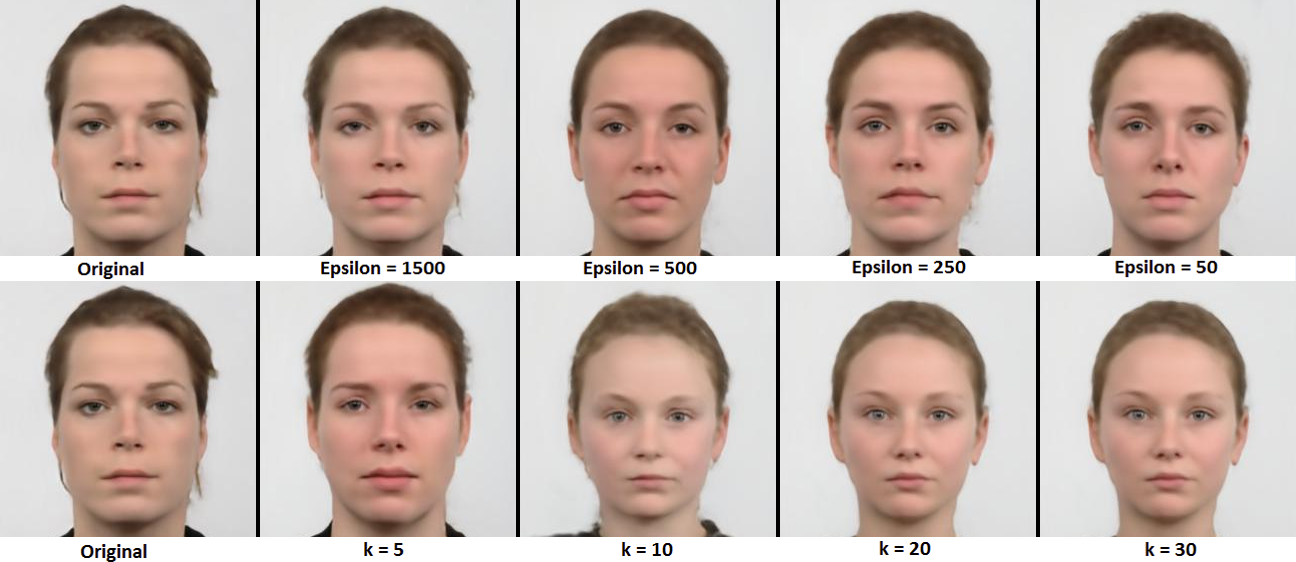}
	\caption{Gender-preserving obfuscation via the generative NN on the RAFD dataset. The top row employs differential privacy and the bottom row employs $k$-same obfuscation.} \label{fig:rafd_gen_images}
\end{figure*}

\begin{figure*}
	\centering
	\includegraphics[width=1\textwidth]{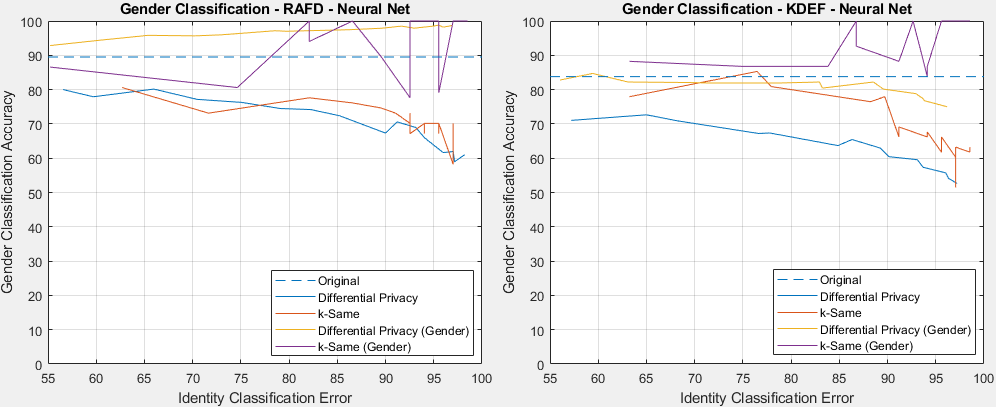}
	\caption{Gender classification accuracy for the methods of obfuscation}
	\label{fig:utility}
\end{figure*}

We note that very recent work \cite{75} has shown the potential for learned representations of high-level features in images to leak sensitive information that can be exploited for unintended inferences. This opens the door to potential leakage of facial identity via features released in an unobfuscated format. Yet, this has been studied in the context of image representations extracted from network layers prior to the final classification output of the model. Such layers retain relatively rich representations of the input which may be used for multiple purposes. We expect that any leakage associated with the extremely coarse features we leave unobfuscated (e.g., gender as a binary input), would have a negligible impact on re-identification risk beyond the explicit revelation (e.g., narrowing candidates based on their gender). To provide some empirical evidence to this effect, we have plotted both the basic and gender-preserving models of the generative NN in the identity classification graphs shown in Figure \ref{fig:privacy}. In all cases, the plot for the gender-preserving model demonstrates only minor deviation from the original model. This deviation is much more likely due to the explicit depiction of gender than the leakage of any additional information.

\begin{figure*} 
	\centering
	\includegraphics[width=0.9\textwidth]{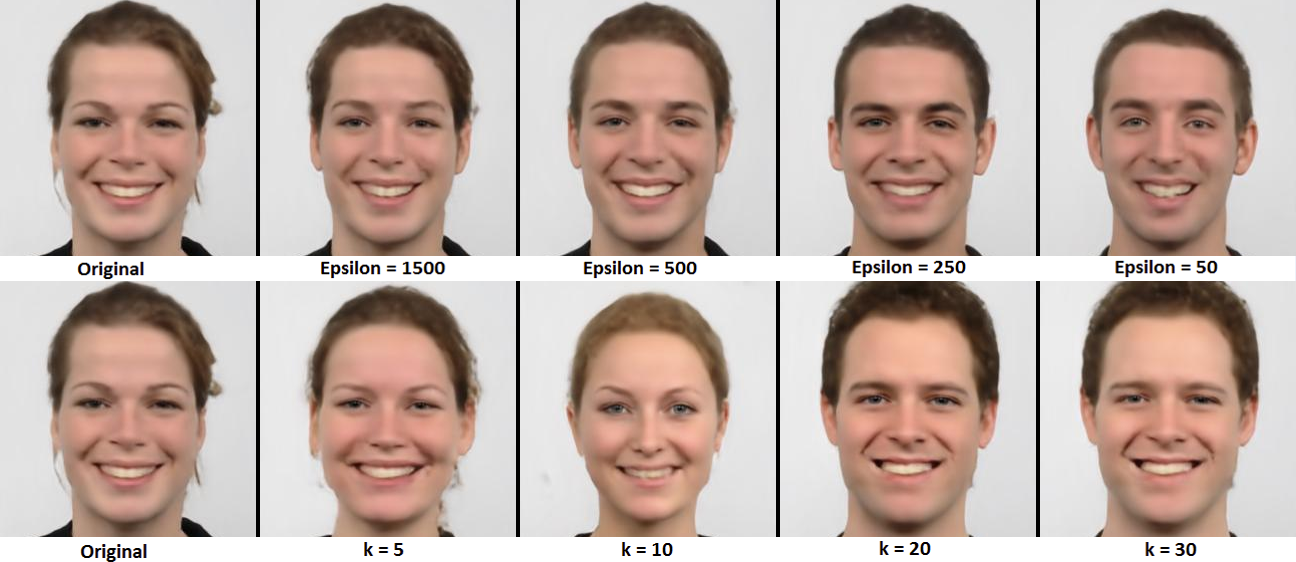}
	\caption{Expression-preserving obfuscation via the generative NN on the RAFD dataset. The top row employs differential privacy and the bottom row employs $k$-same obfuscation.} \label{fig:rafd_em_images}
\end{figure*}

\begin{figure*}
	\centering
	\includegraphics[width=1\textwidth]{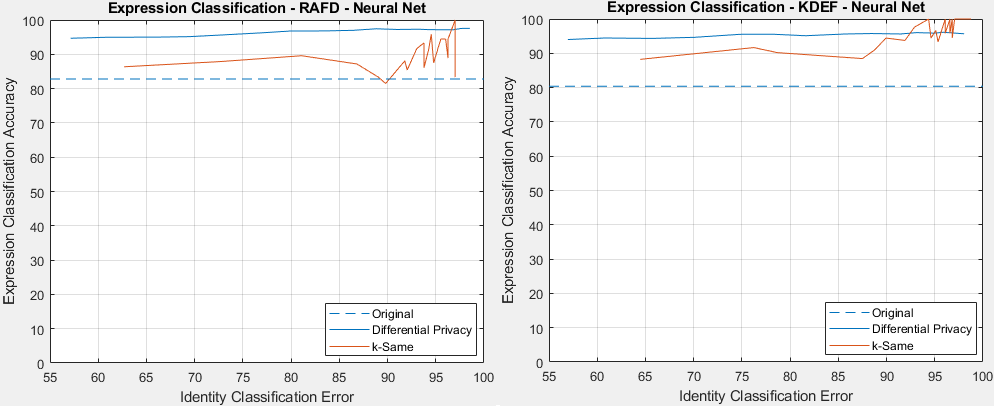}
	\caption{Facial expression classification accuracy for the methods of obfuscation}
	\label{fig:utility_emotion}
\end{figure*}

The results of the gender classification comparisons are shown in Figure \ref{fig:utility}. From these results, we see that the basic models for differential privacy and $k$-same obfuscation suffer a degradation in classification accuracy as the level of privacy is	strengthened (i.e., as identity classification error increases). Comparing the gender-preserved models to their basic counterparts, we see a large improvement in the classification accuracy, suggesting that this is an effective approach for the preservation of specific properties in the obfuscated output. In some cases, the classification accuracy of the obfuscated images has surpassed that of the original data. This is a result of the explicit specification of gender labels in the network input which can lead to obfuscated identities that more prominently display these features.

Another common task is the detection of facial expressions in images \cite{64}. To compare utility in this context, we measure the classification accuracy for facial expressions by using pre-trained neural network models \cite{56} intended for this task. For these experiments, we perform classification for six expressions: happiness, sadness, surprise, disgust, fear and a neutral expression. As the pre-trained models have different strengths and weaknesses with respect to their abilities to classify different expressions, we have employed them in combination to achieve high levels of classification accuracy on our data. We have found that the best results are achieved by first using the Model-4 \cite{56} architecture trained on the RAFDB \cite{65} dataset to detect disgust, followed by the Model-4 architecture trained on the SFEW \cite{66} dataset to detect fear. If neither of these expressions were detected, we then take the sum of the two vectors of model predictions and select the highest prediction as the detected expression. As the models are not trained to detect contempt and give very poor accuracies for detection of anger in our datasets, we exclude these two expressions from our experiments.

For each expression, we generate a full set of obfuscated images as described in the experiments for re-identification risk (Section \ref{rec:reidentification}), providing the expression label as input to the network in order to apply the chosen expression to the obfuscated output. An example of expression-preserving output is shown in Figure \ref{fig:rafd_em_images}. We measure both identity classification accuracy and expression classification accuracy as the average over all six expressions for each privacy parameter used in both differential privacy and $k$-same obfuscation. We then plot the expression classification accuracy as a function of the identity classification accuracy in order to once again abstract away from the privacy parameters and compare the utility of the two methods of obfuscation. The results are shown in Figure \ref{fig:utility_emotion}. As with the experiments on gender classification accuracy, when the chosen attributes are expressly preserved by the network, we observe classification accuracies that exceed the baseline accuracies measured on the unobfuscated datasets. This is again likely due to the network displaying the relevant features more prominently than the original images, making the task of the classification network easier. In these experiments, differential privacy generally shows better utility than $k$-same obfuscation.

Across both utility experiments, the overall comparison between the information preserved in differential privacy and $k$-same obfuscation appears to be inconclusive in these results. Some experiments show better results for differential privacy while other experiments show better results for $k$-same obfuscation. The $k$-same results are also more difficult to assess given the sporadic nature of the plots. This is likely due to changes in clusters between each level of obfuscation which can greatly impact classification accuracy. It is clear that the utility is also data-dependent given the variations in the results seen on the two datasets. Notably, many subjects in the KDEF dataset, including some males, have long hair whereas all subjects in the RAFD dataset have short hair. The males with long hair in KDEF may have contributed to the lower gender classification accuracy.

\subsection{Parrot Attacks}

Methods of $k$-same obfuscation are resistant to parrot attacks as a direct consequence of the process of obfuscation. Differential privacy, due to its resistance to breaches via post-processing of obfuscated data, similarly possesses a theoretical resistance to parrot attacks. Yet, the implications of parrot attacks in practical terms are less clear. A standard classification network applied to images obfuscated via differential privacy is unlikely to perform as well as it could if it exploited public knowledge about how randomization mechanisms function. An attacker could instead train a classification network for facial recognition using instances of (differentially private) obfuscated images as the training set. In this way, the network might achieve higher classification accuracy than a network trained on unobfuscated instances as it has learned to better identify features in the presence of noise. Yet, differential privacy is a stochastic method of obfuscation, so beyond the presence of noise and its approximate magnitude, there is little the network can learn in terms of predictability of differentially private output. We therefore hypothesize that a sufficiently high degree of noise induced by a differentially private mechanism can render the practical implications of a parrot attack negligible.

To observe the degree to which re-identification risk is impacted on differentially private output by parrot attacks, we have trained the VGGFace network for classification of obfuscated instances at specific privacy parameter ($\epsilon$) values. This requires training a separate model for each privacy parameter value on each dataset. We then compare the classification accuracy achieved by the parrot attacks to the accuracy of the models trained on the unobfuscated instances (Figure \ref{fig:parrot}). At high values of $\epsilon$, there is an increase in classification accuracy for the parrot attacks as the network has become better suited to ignoring small amounts of noise. However, as the value of $\epsilon$ decreases, the gap rapidly closes and the trend reverses (at roughly $\epsilon = 300$ for the datasets we have employed), with the parrot attack showing lower classification accuracy than the model trained on unobfuscated data. This is likely due to the higher magnitude of noise destroying many of the useful patterns that the network otherwise learns in the training data. As a result, we expect that for reasonable configurations of privacy parameters that would be used in practice, parrot attacks would provide little, if any, advantage to an attacker.

\begin{figure*}
	\centering
	\includegraphics[width=1\textwidth]{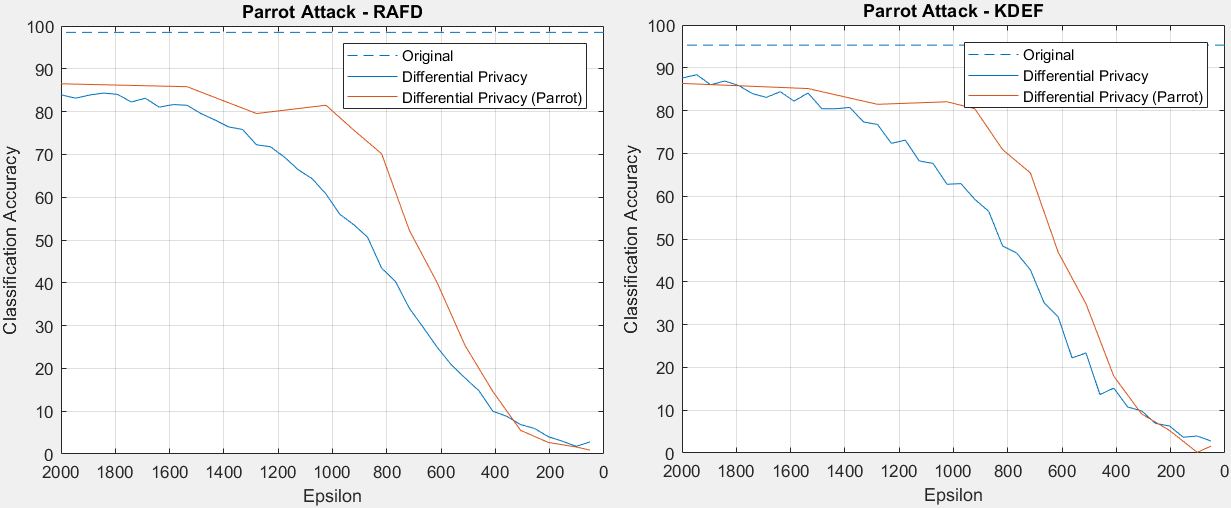}
	\caption{Identity classification accuracy for a parrot attack}
	\label{fig:parrot}
\end{figure*}

\begin{figure*}
	\centering
	\includegraphics[width=1\textwidth]{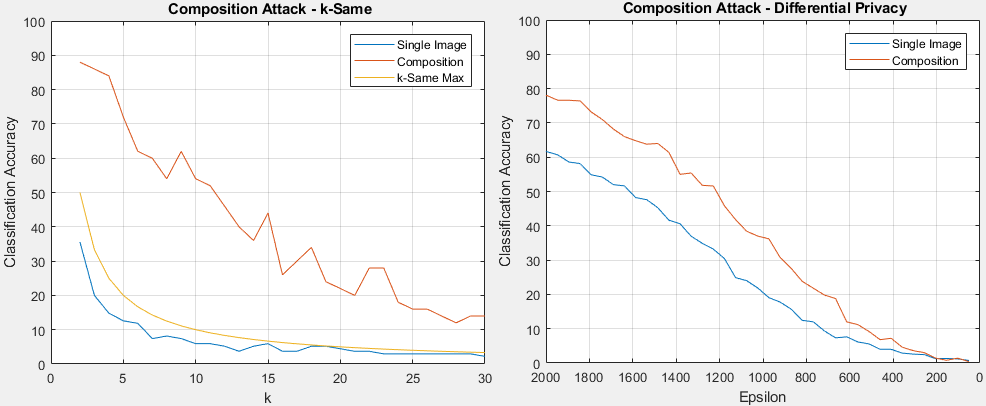}
	\caption{Identity classification accuracy for a composition attack}
	\label{fig:composition}
\end{figure*}

\subsection{Composition Attacks}

One of the key advantages that differential privacy provides over $k$-anonymity is the property of secure composition, which ensures that the privacy guarantee is never violated in the scenario of uncoordinated releases of sensitive data. To demonstrate the resilience of our proposed method of obfuscation against composition attacks, we simulate such an attack and measure the identity classification accuracy for both our implementation of differential privacy and $k$-same obfuscation.

We consider a scenario in which an image is uploaded to two different platforms, each of which apply facial obfuscation in an uncoordinated manner. Through observation of the non-obfuscated portions of the image (e.g., the background of the image), an attacker could determine that the two obfuscated images originally depicted the same individual, enabling them to perform a composition attack. To simulate this, we split a dataset into two subsets of equal size. We select the subsets such that they have a non-empty intersection but also contain many identities that are not found in the other subset. We then train a generative network model on each of the subsets of individuals. This provides us with two models which have some identities in common but which will obfuscate those identities in different ways due to having been trained on differing subsets. These models represent two organizations which will perform obfuscation in an uncoordinated manner. We then use the models to obfuscate only the identities that are present in both subsets. The two obfuscated images produced for each identity represent the two images that an attacker would examine in a composition attack.

As the RAFD and KDEF datasets are too small to achieve any reasonable diversity of identities once they are further reduced to subsets, we combine both datasets together and draw the subsets from the combined dataset. To mitigate bias in identity classification from differences in the controlled settings of the two original datasets, we first crop all images to capture only the faces and then adjust the saturation and contrast of the images to match more closely. From the combined dataset of 135 individuals, we select 50 individuals to be shared across both subsets and split the remaining individuals evenly between them. We match the male to female ratios of both subsets to the full set of images but otherwise select the individuals for each subset at random.

To perform a composition attack on a pair of obfuscated images, we first provide both images as input to the VGGFace network and add together the two resultant vectors of prediction values for the full set of identities. For $k$-same obfuscation, we then take the intersection of the identities from each cluster used by the two models during obfuscation and select the identity within this intersection that has the highest prediction value from the summed vectors. In practice, an attacker will not necessarily know the exact identities used in each cluster with certainty but could likely determine them with reasonable accuracy by taking the top $k$ predictions from their facial recognition network on each of the obfuscated images. By using the exact identities in our experiments, we effectively test the worst-case scenario.

When launching a composition attack against differentially private output, the attacker no longer has the concept of clusters of identities to use to their advantage. The original image may have depicted any of the potential identities. The attacker may still use the two images to increase the accuracy of their prediction but the additional information does not provide them with any means to violate the differential privacy guarantee. To simulate an attacker using the additional information in this context, we select the highest prediction value from the summed vectors of predictions for each pair of images. Since we generate ten instances of obfuscated output for each identity, we average the prediction accuracy over all ten pairs of images for each identity.

The results are shown in Figure \ref{fig:composition}, comparing the re-identification risk in each method of obfuscation for a single obfuscated image to the risk from a pair of images on which a composition attack has been performed. With the $k$-same obfuscation, we observe a significant increase in re-identification risk which greatly exceeds the theoretical maximum value of $\frac{1}{k}$. This clearly demonstrates a violation of the $k$-same privacy guarantee. With differential privacy, we also observe an increase in re-identification risk, as would be expected due to the additional information provided to the attacker, however, the property of secure composition ensures that the privacy guarantee is preserved. Furthermore, the gap between the single image re-identification risk and composition risk is less significant on the differentially private images than the $k$-same images, and becomes marginal at high levels of privacy (e.g., when $\epsilon < 300$ in Figure \ref{fig:composition}). Since privacy parameters that are conducive to a high level of privacy would typically be selected in most realistic scenarios, we expect that the degradation in privacy due to multiple instances of differentially private output being distributed in practice would be minimal.

\subsection{Privacy-Utility Trade-offs in Pixel-Space}

In this section, we compare our proposed use of a pixel-space exponential mechanism to an existing implementation \cite{5} of a Laplace mechanism in pixel-space. In these experiments, we employ the FaceScrub dataset \cite{85} to reflect the ability of pixel-space obfuscation to handle diverse image content. FaceScrub consists of a collection of facial images spanning roughly 500 individuals in diverse conditions with respect to background content, lighting, pose, etc. We resize all images to $128 \times 128$ pixels and randomly select ten images per individual to act as a training dataset for a facial classification network and three images per individual to act as the testing set. Given that differentially private obfuscation is a stochastic process, we produce three obfuscated instances per test image resulting in a total of nine obfuscated images per identity. We re-iterate that pixel-space obfuscation requires neither training data nor a generative model for the data. The use of training data in our experiments is strictly to train a classification network in order to measure re-identification risk and report on the performance of the obfuscation mechanisms.

The work in \cite{5} proposes the combination of pixelization with differential privacy in order to better manage the privacy budget. We therefore incorporate this into our experiments by using three different pixelization grid sizes: $4 \times 4$, $8 \times 8$, and $16 \times 16$. As it is known that blurring a pixelized image can improve the ability of humans to recognize the image content \cite{27}, we additionally test versions of the images that have been blurred as a post-processing step after obfuscation has been applied. We blur the images using a Gaussian kernel with a standard deviation of 1. Recall that post-processing cannot impact the differential privacy guarantee \cite{30}. However, this does not preclude the possibility of an impact on both re-identification risk and utility.

\subsubsection{Privacy Budget Comparisons}

We begin by examining the impact of pixelization and blurring on the obfuscated output. The three pixelization settings combined with a boolean option for blurring results in six potential configurations for each of mechanisms. We first compare the variants of the exponential mechanism separately from the variants of the Laplace mechanism in order to clearly observe the impact of these settings within the same class of mechanism. The results of this comparison are shown in Figure \ref{fig:epsilon_ssim}. Here we plot SSIM as a function of the privacy budget $\epsilon$ (i.e., the composition of the privacy parameters over all applications of the mechanism needed to obfuscate an image). The SSIM values are calculated as the average score over all obfuscated instances. Lower values on the x-axis represent a stronger privacy guarantee and higher values on the y-axis represent better utility. We additionally show examples of obfuscated images in Figure \ref{fig:pixel_examples}.

\begin{figure*}
	\centering
	\includegraphics[width=0.95\textwidth]{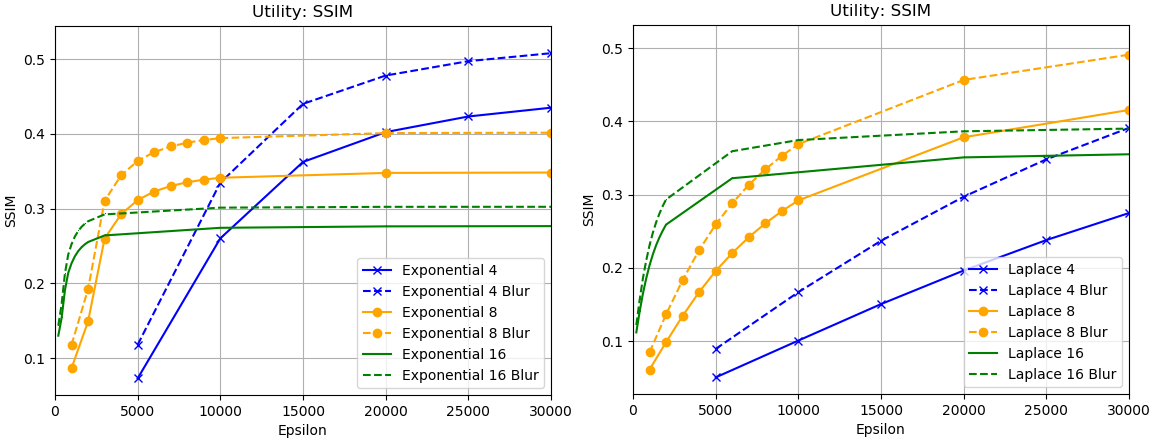}
	\caption{A comparison of the SSIM achieved by the mechanisms over a wide range of privacy budgets. The left graph compares variants of the exponential mechanism and the right graph compares variants of the Laplace mechanism.}
	\label{fig:epsilon_ssim}
\end{figure*}

\begin{figure*}
	\centering
	\includegraphics[width=0.6\textwidth]{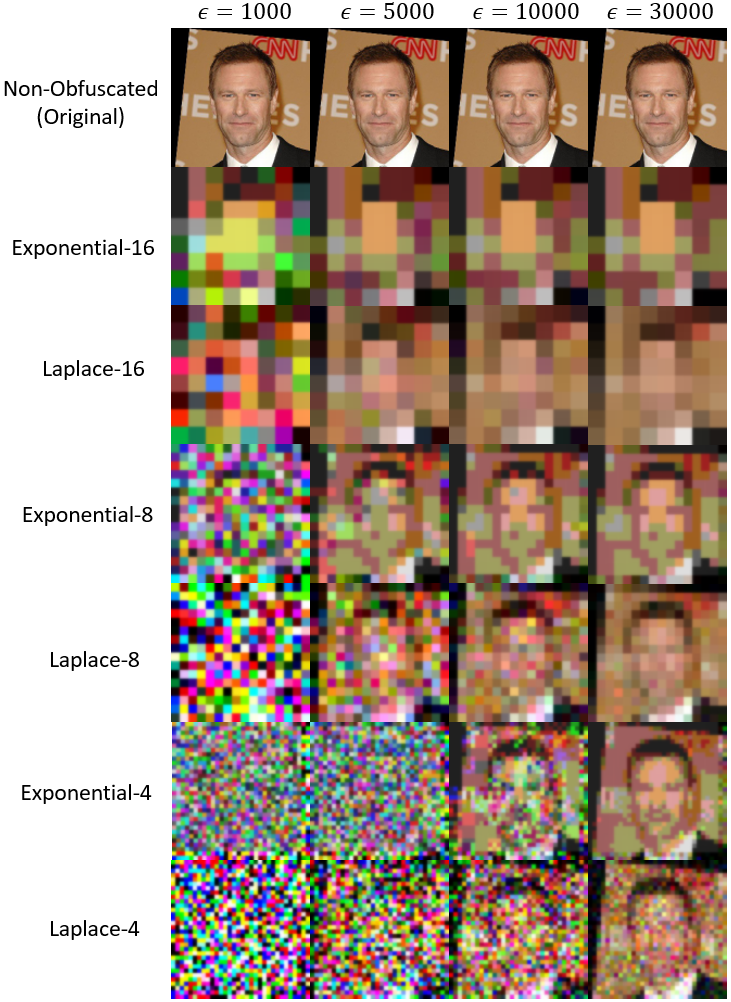}
	\caption{A visual comparison of obfuscated instances produced by the blur variants of both mechanisms at all 3 pixelization settings.}
	\label{fig:pixel_examples}
\end{figure*}

From the results of Figure \ref{fig:epsilon_ssim} we note two main trends that hold for both the exponential and the Laplace mechanism. First, for any given pixelization setting, the use of blurring always offers a higher level of utility than the non-blurred configuration. This can be seen by the dashed plots representing the blur variants which universally result in higher SSIM than their non-blur counterparts shown by the solid line plots. The second common trend is that stronger pixelization (i.e., larger pixelization grid sizes) provides higher SSIM for strong levels of privacy (i.e., low values of $\epsilon$), while weaker pixelization provides higher SSIM for weak levels of privacy. The better performance of strong pixelization for small values of $\epsilon$ follows from the usefulness of pixelization as a means to reduce the query sensitivity, leading to better management of the privacy budget. The reversal of this trend for large values of $\epsilon$ is a result of the minimal amount of noise that is added by the mechanisms for such weak privacy parameters. Due to this, the predominant modification to the images is the pixelization rather than the application of noise, leading to worse utility in the stronger levels of pixelization. Similar results can be seen in Figure \ref{fig:epsilon_mse} of Appendix B using the mean squared error (MSE) of pixel intensities as an alternate measure of utility. Note that in the case of MSE, lower values indicate better utility.

Given that the intent is to provide a meaningful level of privacy through differentially private obfuscation, these large privacy budgets are of little practical value. This can be easily confirmed by visual inspection of the obfuscated examples at high values of $\epsilon$ in Figure \ref{fig:pixel_examples}. Consequently, strong pixelization with small values of $\epsilon$ should be used in practice to achieve meaningful levels of privacy while attaining the best trade-off with respect to utility. Furthermore, the blurring operation appears to always be beneficial with respect to utility in the obfuscated output.

We next turn to a comparison between the exponential mechanism and the Laplace mechanism. Using only the blurred variants, we compare the two mechanisms at all three levels of pixelization. The results are shown in Figure \ref{fig:epsilon_ssim_both}-a. For any given privacy budget $\epsilon$, the highest plotted mechanism on the y-axis reflects the best performance. The exponential mechanisms, plotted as dashed lines, almost exclusively make up the upper envelope of the plots. This demonstrates a consistently stronger performance from our proposed mechanism in comparison to the Laplace mechanism. Furthermore, since we are primarily interested in the privacy/utility trade-off for low values of $\epsilon$ which provide a meaningful privacy guarantee, we show a zoomed in comparison of the two mechanisms using a pixelization setting of 16 in Figure \ref{fig:epsilon_ssim_both}-b. This comparison shows a notable improvement in SSIM from the exponential mechanism in comparison to the Laplace mechanism for strong levels of privacy (i.e., $\epsilon < 1000$). Analogous comparisons are shown for MSE in Figure \ref{fig:epsilon_mse_both} of Appendix B.

\begin{figure*}
	\centering
	\includegraphics[width=0.95\textwidth]{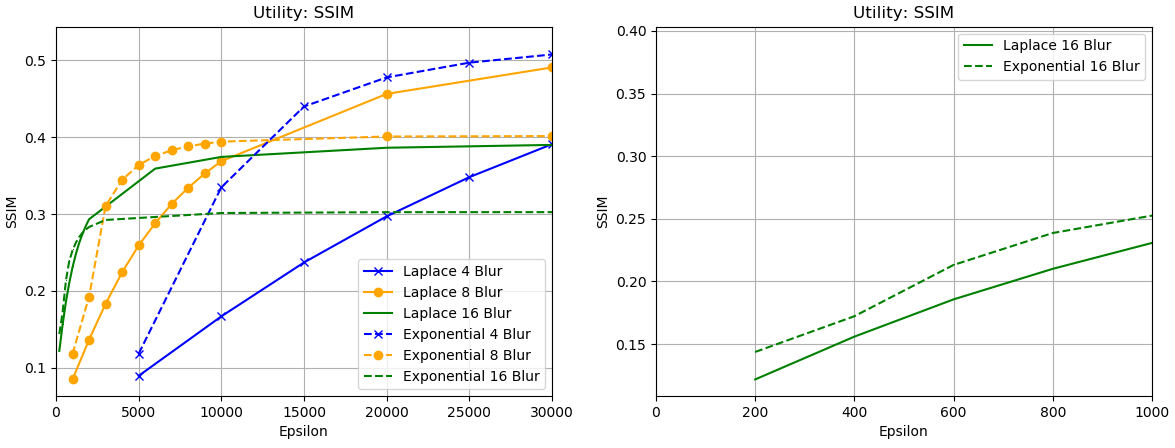}
	\caption{A comparison between the blur variants of the exponential and Laplace mechanism with respect to utility measured using SSIM. The left graph compares all three pixelization settings for the mechanisms and the right graph focuses on the mechanisms using a pixelization grid of size 16 for efficient use of low privacy budgets.}
	\label{fig:epsilon_ssim_both}
\end{figure*}

\begin{figure*}
	\centering
	\includegraphics[width=0.95\textwidth]{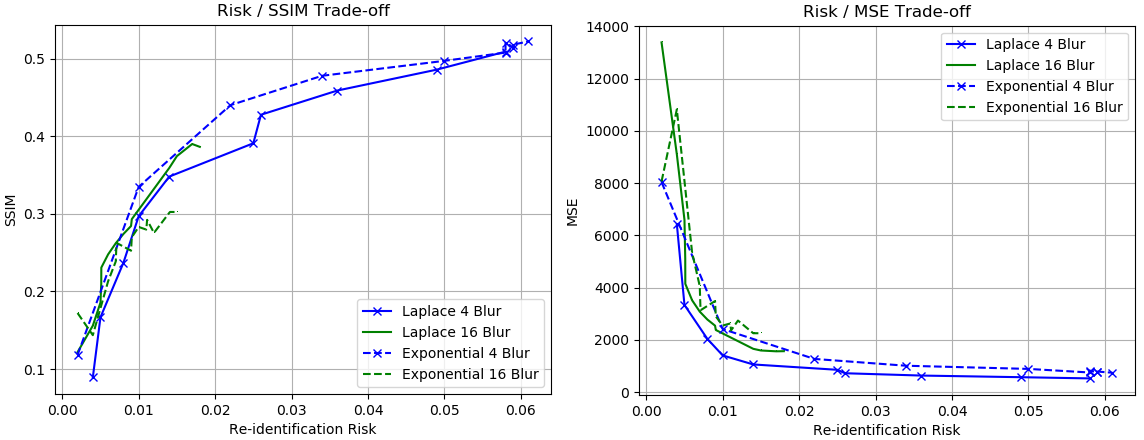}
	\caption{A comparison of the re-identification risk and utility trade-off of the mechanisms. Re-identification risk is measured using the average classification accuracy of a facial identity classification network.}
	\label{fig:pixel_risk}
\end{figure*}

Despite the better performance achieved via the use of strong pixelization and blurring for low values of $\epsilon$, the scale on which $\epsilon$ is interpreted remains quite large. Yet, given the obfuscated examples of Figure \ref{fig:pixel_examples}, it is clear that these values of $\epsilon$ provide a strong level of obfuscation. These results suggest that the interpretation of the privacy guarantee using concepts of adjacency and query sensitivity may not be appropriate for pixel-space obfuscation. As with our proposed approach to obfuscation via generative models, the privacy guarantee is likely better interpreted in pixel-space using a distance-based generalization of differential privacy. We leave further investigation of this topic as an open question.

\subsubsection{Risk/Utility Trade-offs}

As a final comparison of the mechanisms, we use a facial identity classification network to measure re-identification risk. To this end, we employ FaceNet \cite{86}, a network which has been demonstrated to provide high levels of identity classification accuracy even when presented with low resolution or otherwise poor quality images. This property makes it a good candidate for re-identification of images obfuscated in pixel space, which are subject to a reduction in resolution due to pixelization and heavy distortion due to the addition of noise. For reference, we provide the network architecture in Table \ref{tab:facenet} of Appendix A. For each combination of a mechanism type, pixelization setting, and privacy budget, we train the network using the training partition of the dataset subject to obfuscation via the selected mechanism configuration. As with the testing data, we produce three obfuscated instances per image, leading to a total of 30 training instances per identity. This allows us to use the classification network to launch a parrot attack on the obfuscated data.

Results are shown in Figure \ref{fig:pixel_risk}. We plot utility as a function of re-identification risk which is measured as the average identity classification accuracy over the obfuscated testing images. A desirable privacy/utility trade-off is reflected by low re-identification risk and high utility. In Figure \ref{fig:pixel_risk}-a, we measure utility as SSIM where high values on the y-axis indicate higher utility and in Figure \ref{fig:pixel_risk}-b, we measure utility using MSE where low values on the y-axis indicate higher utility. In both graphs, we show the mechanisms using pixelization settings of 4 and 16. We omit the pixelization setting of 8 for better visual clarity due to the clutter caused by the closeness of the plots. Note that the plots for the variants using a pixelization setting of 16 all remain beneath 2\% re-identification risk. This is due to the strong level of pixelization which prevents effective re-identification even in the absence of noise added by the mechanisms.

Interestingly, the variants using a pixelization setting of 4 provide the best utility, with the exponential mechanism providing the best performance with respect to SSIM and the Laplace mechanism providing the best performance with respect to MSE. However, as we have noted in the previous experiments, the variants using a pixelization setting of 4 have a worse trade-off between the privacy budget and the level of utility. Due to this, the better performance observed in Figure \ref{fig:pixel_risk} should be considered with caution as the higher values of $\epsilon$ imply worse theoretical properties of privacy protection despite the low re-identification risk. It is, for example, possible that the use of different classification networks which exploit different aspects of the obfuscated images may show different trends in the trade-off between re-identification risk and utility. We leave further investigation into the measure of re-identification risk in practice as future work.

\subsubsection{Summary of Results}

Based on the results of our experimental comparisons, we recommend the use of the exponential mechanism along with a strong level of pixelization (e.g., a pixelization grid size of 16) to achieve high levels of utility in both SSIM and MSE while using a low privacy budget in order to enforce a strong level of privacy. Furthermore, blurring can be applied to the obfuscated results to improve the level of utility. If one is strictly interested in a practical interpretation of privacy as re-identification risk, the use of a weaker level of pixelization appears to have the potential to provide a better risk/utility trade-off. However, we caution readers that this comes at the cost the strong theoretical properties of privacy protection that would otherwise be achieved via a low value of $\epsilon$.

\section{Conclusions and Future Work}

We have studied how to obtain a formalized privacy guarantee for the obfuscation of facial images in practice. We have identified shortcomings of the $k$-same privacy guarantee including susceptibilities to background knowledge and composition attacks as well as the awkwardness in the requirement for a gallery of input images. To improve upon this, we have proposed the use of differential privacy in the context of obfuscation applied to generative models for images. We have developed a framework that provides a meaningful privacy guarantee for such models and we have derived the configuration of Laplace mechanism that can achieve this privacy guarantee. Our approach preserves the privacy guarantee in the presence of attackers with background knowledge, provides resistance to composition attacks and removes the requirement for a gallery of input images. We have also proposed the use of a more general mechanism to obfuscate any image directly in pixel-space. This allows for greater versatility in the obfuscation of images.

We have implemented our proposed mechanisms as well as the competing approaches discussed in the paper for experimental comparisons. In our experiments on pixel-space mechanisms, we have demonstrated improvements in measures of visual quality for our proposed use of an exponential mechanism over the Laplace mechanism. In our experiments using generative models, we have implemented both our proposed framework as well as $k$-same obfuscation. Through our comparisons, we have demonstrated the resilience of differential privacy against parrot and composition attacks. Furthermore, we have shown that this application of differential privacy can achieve comparable utility to $k$-same obfuscation. We conclude that the key improvements in the privacy guarantee combined with comparable levels of utility make differential privacy a much more appropriate choice for the obfuscation of facial images.

In this work, we have focused on introducing a framework and examining its theoretical and practical properties for frontal facing images taken in a controlled environment using standard facial image datasets. In future work, this could be extended to studying generative adversarial networks for the generation of detailed images drawn from more complex distributions. Furthermore, we posit that the approximation of unknown identities can be improved upon through the use of auto-encoder style architectures seen in some generative networks.

We also leave as an open problem the potential for improvements to the exponential mechanism for pixel-space obfuscation. Given that our proposed implementation is limited both in the size of the exponential grid and the precision of the pixel intensities, there is room for improvement in the quality of the obfuscated output if one can find more effective methods to handle the computational complexity of the mechanism. Additionally, investigation of SSIM variants for use as the quality function may lead to improved preservation of visual quality in the obfuscated output. For example, the complex wavelet SSIM \cite{87}, which is robust to small changes in translation and scaling, could be explored for use in this context.

\bibliographystyle{ieeetr}
\bibliography{bibliography}

\begin{thebibliography}{10}

\bibitem{38}
A.~{Cavailaro}, ``{Privacy in Video Surveillance [In the Spotlight]},'' {\em
  IEEE Signal Processing Magazine}, vol.~24, no.~2, pp.~168--166, 2007.

\bibitem{37}
T.~Winkler and B.~Rinner, ``{Security and Privacy Protection in Visual Sensor
  Networks: A Survey},'' {\em ACM Comput. Surv.}, vol.~47, no.~1, pp.~1--42,
  2014.

\bibitem{57}
S.~Ribaric, A.~Ariyaeeinia, and N.~Pavesic, ``{De-identification for Privacy
  Protection in Multimedia Content: A Survey},'' {\em Signal Processing: Image
  Communication}, vol.~47, pp.~131 -- 151, 2016.

\bibitem{39}
J.~R. Padilla-L\'{o}pez, A.~A. Chaaraoui, and F.~Fl\'{o}rez-Revuelta, ``{Visual
  Privacy Protection Methods: A Survey},'' {\em Expert Syst. Appl.}, vol.~42,
  no.~9, pp.~4177--4195, 2015.

\bibitem{41}
Google, ``{Google Maps}.''
\newblock Accessed: February 27, 2019.

\bibitem{42}
A.~Frome, G.~Cheung, A.~Abdulkader, M.~Zennaro, B.~Wu, A.~Bissacco, H.~Adam,
  H.~Neven, and L.~Vincent, ``{Large-Scale Privacy Protection in Google Street
  View},'' in {\em IEEE 12th International Conference on Computer Vision},
  pp.~2373--2380, 2009.

\bibitem{34}
K.~Martin and K.~Shilton, ``{Putting Mobile Application Privacy in Context: An
  Empirical Study of User Privacy Expectations for Mobile Devices},'' {\em The
  Information Society}, vol.~32, no.~3, pp.~200--216, 2016.

\bibitem{35}
X.~Hu, D.~Hu, S.~Zheng, W.~Li, F.~Chen, Z.~Shu, and L.~Wang, ``{How People
  Share Digital Images in Social Networks: A Questionnaire-Based Study of
  Privacy Decisions and Access Control},'' {\em Multimedia Tools and
  Applications}, vol.~77, no.~14, pp.~18163--18185, 2018.

\bibitem{19}
O.~M. Parkhi, A.~Vedaldi, and A.~Zisserman, ``{Deep Face Recognition},'' in
  {\em British Machine Vision Conference}, 2015.

\bibitem{32}
Y.~{Li}, N.~{Vishwamitra}, B.~P. {Knijnenburg}, H.~{Hu}, and K.~{Caine},
  ``{Blur vs. Block: Investigating the Effectiveness of Privacy-Enhancing
  Obfuscation for Images},'' in {\em IEEE Conference on Computer Vision and
  Pattern Recognition Workshops}, pp.~1343--1351, 2017.

\bibitem{33}
Y.~Li, N.~Vishwamitra, B.~P. Knijnenburg, H.~Hu, and K.~Caine, ``{Effectiveness
  and Users' Experience of Obfuscation As a Privacy-Enhancing Technology for
  Sharing Photos},'' {\em Proc. ACM Hum.-Comput. Interact.}, vol.~1, pp.~1--24,
  2017.

\bibitem{59}
G.~{Cormode}, C.~M. {Procopiuc}, E.~{Shen}, D.~{Srivastava}, and T.~{Yu},
  ``{Empirical Privacy and Empirical Utility of Anonymized Data},'' in {\em
  IEEE 29th International Conference on Data Engineering Workshops},
  pp.~77--82, 2013.

\bibitem{60}
M.~M. {Almasi}, T.~R. {Siddiqui}, N.~{Mohammed}, and H.~{Hemmati}, ``{The
  Risk-Utility Tradeoff for Data Privacy Models},'' in {\em 8th IFIP
  International Conference on New Technologies, Mobility and Security},
  pp.~1--5, 2016.

\bibitem{43}
S.~E. Hudson and I.~Smith, ``{Techniques for Addressing Fundamental Privacy and
  Disruption Tradeoffs in Awareness Support Systems},'' in {\em ACM Conference
  on Computer Supported Cooperative Work}, pp.~248--257, 1996.

\bibitem{49}
J.~{Kröckel} and F.~{Bodendorf}, ``{Customer Tracking and Tracing Data as a
  Basis for Service Innovations at the Point of Sale},'' in {\em Annual SRII
  Global Conference}, pp.~691--696, 2012.

\bibitem{50}
X.~Liu, N.~Krahnstoever, T.~Yu, and P.~Tu, ``{What Are Customers Looking
  at?},'' in {\em IEEE Conference on Advanced Video and Signal Based
  Surveillance}, pp.~405--410, 2007.

\bibitem{51}
F.~Anwar, I.~Petrounias, T.~Morris, and V.~Kodogiannis, ``{Mining Anomalous
  Events Against Frequent Sequences in Surveillance Videos from Commercial
  Environments},'' {\em Expert Syst. Appl.}, vol.~39, no.~4, pp.~4511--4531,
  2012.

\bibitem{52}
P.~L. {Venetianer}, Z.~{Zhang}, A.~{Scanlon}, Y.~{Hu}, and A.~J. {Lipton},
  ``{Video Verification of Point of Sale Transactions},'' in {\em IEEE
  Conference on Advanced Video and Signal Based Surveillance}, pp.~411--416,
  2007.

\bibitem{1}
E.~M. {Newton}, L.~{Sweeney}, and B.~{Malin}, ``{Preserving Privacy by
  De-Identifying Face Images},'' {\em IEEE Transactions on Knowledge and Data
  Engineering}, vol.~17, no.~2, pp.~232--243, 2005.

\bibitem{36}
S.~R. Ganta, S.~P. Kasiviswanathan, and A.~Smith, ``{Composition Attacks and
  Auxiliary Information in Data Privacy},'' in {\em 14th ACM SIGKDD
  International Conference on Knowledge Discovery and Data Mining},
  pp.~265--273, 2008.

\bibitem{28}
L.~Harmon and B.~Julesz, ``{Masking in Visual Recognition: Effects of
  Two-Dimensional Filtered Noise},'' {\em Science}, vol.~180, no.~4091,
  pp.~1194--1197, 1973.

\bibitem{27}
L.~Harmon, ``{The Recognition of Faces},'' {\em Scientific American}, vol.~229,
  no.~5, pp.~71--82, 1973.

\bibitem{12}
G.~{Letournel}, A.~{Bugeau}, V.~. {Ta}, and J.~. {Domenger}, ``{Face
  De-Identification with Expressions Preservation},'' in {\em IEEE
  International Conference on Image Processing}, pp.~4366--4370, 2015.

\bibitem{13}
P.~{Korshunov} and T.~{Ebrahimi}, ``{Using Warping for Privacy Protection in
  Video Surveillance},'' in {\em 18th International Conference on Digital
  Signal Processing}, pp.~1--6, 2013.

\bibitem{14}
P.~{Korshunov} and T.~{Ebrahimi}, ``{Using Face Morphing to Protect Privacy},''
  in {\em 10th IEEE International Conference on Advanced Video and Signal Based
  Surveillance}, pp.~208--213, 2013.

\bibitem{8}
D.~Bitouk, N.~Kumar, S.~Dhillon, P.~Belhumeur, and S.~K. Nayar, ``{Face
  Swapping: Automatically Replacing Faces in Photographs},'' {\em ACM Trans.
  Graph.}, vol.~27, no.~3, pp.~39:1--39:8, 2008.

\bibitem{6}
S.~Mosaddegh, L.~Simon, and F.~Jurie, ``{Photorealistic Face De-Identification
  by Aggregating Donors' Face Components},'' in {\em Asian Conference on
  Computer Vision}, pp.~159--174, 2015.

\bibitem{47}
A.~Brand and J.~A. Lal, ``{European Best Practice for Quality Assurance,
  Provision and Use of Genome-based Information and Technologies},'' {\em Drug
  Metabol Drug Interact}, vol.~27, pp.~177--82, 2012.

\bibitem{48}
{U.S. Department of Health \& Human Services}, ``{Guidance Regarding Methods
  for De-identification of Protected Health Information in Accordance with the
  Health Insurance Portability and Accountability Act (HIPAA) Privacy Rule},''
  2015.
\newblock Accessed: February 9, 2018.

\bibitem{31}
P.~Samarati and L.~Sweeney, ``{Protecting Privacy when Disclosing Information:
  k-Anonymity and its Enforcement through Generalization and Suppression},''
  tech. rep., Computer Science Laboratory, {SRI} International, 1998.

\bibitem{2}
R.~{Gross}, L.~{Sweeney}, F.~{de la Torre}, and S.~{Baker}, ``{Model-Based Face
  De-Identification},'' in {\em Computer Vision and Pattern Recognition
  Workshop}, pp.~161--161, 2006.

\bibitem{29}
T.~F. Cootes, G.~J. Edwards, and C.~J. Taylor, ``{Active Appearance Models},''
  {\em IEEE Trans. Pattern Anal. Mach. Intell.}, vol.~23, no.~6, pp.~681--685,
  2001.

\bibitem{58}
L.~{Meng} and Z.~{Sun}, ``{Face De-identification with Perfect Privacy
  Protection},'' in {\em 37th International Convention on Information and
  Communication Technology, Electronics and Microelectronics}, pp.~1234--1239,
  2014.

\bibitem{3}
H.~{Chi} and Y.~H. {Hu}, ``{Face De-Identification Using Facial Identity
  Preserving Features},'' in {\em IEEE Global Conference on Signal and
  Information Processing}, pp.~586--590, 2015.

\bibitem{4}
B.~{Meden}, Z.~{Emersic}, V.~{Struc}, and P.~{Peer}, ``{k-Same-Net:
  Neural-Network-Based Face Deidentification},'' in {\em International
  Conference and Workshop on Bioinspired Intelligence}, pp.~1--7, 2017.

\bibitem{26}
A.~{Dosovitskiy}, J.~T. {Springenberg}, M.~{Tatarchenko}, and T.~{Brox},
  ``{Learning to Generate Chairs, Tables and Cars with Convolutional
  Networks},'' {\em IEEE Transactions on Pattern Analysis and Machine
  Intelligence}, vol.~39, no.~4, pp.~692--705, 2017.

\bibitem{9}
R.~Gross, E.~Airoldi, B.~Malin, and L.~Sweeney, ``{Integrating Utility into
  Face De-identification},'' in {\em Privacy Enhancing Technologies},
  pp.~227--242, 2006.

\bibitem{11}
L.~{Du}, M.~{Yi}, E.~{Blasch}, and H.~{Ling}, ``{GARP-Face: Balancing Privacy
  Protection and Utility Preservation in Face De-Identification},'' in {\em
  IEEE International Joint Conference on Biometrics}, pp.~1--8, 2014.

\bibitem{7}
T.~{Sim} and L.~{Zhang}, ``{Controllable Face Privacy},'' in {\em 11th IEEE
  International Conference and Workshops on Automatic Face and Gesture
  Recognition}, vol.~04, pp.~1--8, 2015.

\bibitem{10}
L.~{Meng}, Z.~{Sun}, A.~{Ariyaeeinia}, and K.~L. {Bennett}, ``{Retaining
  Expressions on De-Identified Faces},'' in {\em 37th International Convention
  on Information and Communication Technology, Electronics and
  Microelectronics}, pp.~1252--1257, 2014.

\bibitem{5}
L.~Fan, ``{Image Pixelization with Differential Privacy},'' in {\em Data and
  Applications Security and Privacy}, pp.~148--162, 2018.

\bibitem{74}
L.~{Fan}, ``Practical image obfuscation with provable privacy,'' in {\em IEEE
  International Conference on Multimedia and Expo}, pp.~784--789, 2019.

\bibitem{84}
W.~L. Croft, J.~Sack, and W.~Shi, ``Differentially private obfuscation of
  facial images,'' in {\em Machine Learning and Knowledge Extraction},
  pp.~229--249, 2019.

\bibitem{76}
R.~Shokri, M.~Stronati, C.~Song, and V.~Shmatikov, ``{Membership Inference
  Attacks Against Machine Learning Models},'' in {\em {IEEE} Symposium on
  Security and Privacy}, pp.~3--18, 2017.

\bibitem{77}
B.~Hitaj, G.~Ateniese, and F.~P{\'{e}}rez{-}Cruz, ``{Deep Models Under the
  {GAN:} Information Leakage from Collaborative Deep Learning},'' in {\em {ACM}
  {SIGSAC} Conference on Computer and Communications Security}, pp.~603--618,
  2017.

\bibitem{78}
L.~Melis, C.~Song, E.~D. Cristofaro, and V.~Shmatikov, ``{Exploiting Unintended
  Feature Leakage in Collaborative Learning},'' in {\em {IEEE} Symposium on
  Security and Privacy}, pp.~691--706, 2019.

\bibitem{79}
M.~Fredrikson, S.~Jha, and T.~Ristenpart, ``{Model Inversion Attacks that
  Exploit Confidence Information and Basic Countermeasures},'' in {\em 22nd
  {ACM} {SIGSAC} Conference on Computer and Communications Security},
  pp.~1322--1333, 2015.

\bibitem{80}
M.~Abadi, A.~Chu, I.~J. Goodfellow, H.~B. McMahan, I.~Mironov, K.~Talwar, and
  L.~Zhang, ``{Deep Learning with Differential Privacy},'' in {\em {ACM}
  {SIGSAC} Conference on Computer and Communications Security}, pp.~308--318,
  2016.

\bibitem{61}
L.~Xie, K.~Lin, S.~Wang, F.~Wang, and J.~Zhou, ``{Differentially Private
  Generative Adversarial Network},'' {\em CoRR}, 2018.

\bibitem{62}
X.~Zhang, S.~Ji, and T.~Wang, ``{Differentially Private Releasing via Deep
  Generative Model},'' {\em CoRR}, 2018.

\bibitem{63}
A.~Triastcyn and B.~Faltings, ``{Generating Differentially Private Datasets
  Using GANs},'' {\em CoRR}, 2018.

\bibitem{81}
P.~C. Roy and V.~N. Boddeti, ``{Mitigating Information Leakage in Image
  Representations: {A} Maximum Entropy Approach},'' in {\em {IEEE} Conference
  on Computer Vision and Pattern Recognition}, pp.~2586--2594, 2019.

\bibitem{82}
J.~Chen, J.~Konrad, and P.~Ishwar, ``{VGAN-Based Image Representation Learning
  for Privacy-Preserving Facial Expression Recognition},'' in {\em {IEEE}
  Conference on Computer Vision and Pattern Recognition Workshops},
  pp.~1570--1579, 2018.

\bibitem{83}
Z.~Ren, Y.~J. Lee, and M.~S. Ryoo, ``{Learning to Anonymize Faces for Privacy
  Preserving Action Detection},'' in {\em European Conference on Computer
  Vision}, pp.~639--655, 2018.

\bibitem{68}
Y.~Wu, F.~Yang, Y.~Xu, and H.~Ling, ``{Privacy-Protective-GAN for Privacy
  Preserving Face De-Identification},'' {\em Journal of Computer Science and
  Technology}, vol.~34, no.~1, pp.~47--60, 2019.

\bibitem{46}
A.~{Basu}, T.~{Nakamura}, S.~{Hidano}, and S.~{Kiyomoto}, ``{k-anonymity: Risks
  and the Reality},'' in {\em IEEE Trustcom/BigDataSE/ISPA}, vol.~1,
  pp.~983--989, 2015.

\bibitem{71}
K.~Nissim and A.~Wood, ``{Is Privacy \emph{Privacy}?},'' {\em The Royal
  Society}, vol.~376, 2018.

\bibitem{72}
K.~Nissim, A.~Bembenek, A.~Wood, M.~Bun, M.~Gaboardi, U.~Gasser, D.~R.
  O{\textquoteright}Brien, and S.~Vadhan, ``{Bridging the Gap between Computer
  Science and Legal Approaches to Privacy},'' {\em Harvard Journal of Law \&
  Technology}, vol.~31, pp.~687--780, 2018.

\bibitem{73}
{Rachel Cummings and Deven Desai}, ``{The Role of Differential Privacy in GDPR
  Compliance},'' 2018.
\newblock Accessed: October 2, 2019.

\bibitem{22}
C.~Dwork, ``{Differential Privacy},'' in {\em 33rd International Conference on
  Automata, Languages and Programming - Volume Part II}, pp.~1--12, 2006.

\bibitem{30}
C.~Dwork and A.~Roth, ``{The Algorithmic Foundations of Differential
  Privacy},'' {\em Found. Trends Theor. Comput. Sci.}, vol.~9, no.~3\&\#8211;4,
  pp.~211--407, 2014.

\bibitem{23}
K.~Chatzikokolakis, M.~E. Andr{\'e}s, N.~E. Bordenabe, and C.~Palamidessi,
  ``{Broadening the Scope of Differential Privacy Using Metrics},'' in {\em
  Privacy Enhancing Technologies}, pp.~82--102, 2013.

\bibitem{53}
M.~E. Andr{\'e}s, N.~E. Bordenabe, K.~Chatzikokolakis, and C.~Palamidessi,
  ``{Geo-indistinguishability: Differential Privacy for Location-based
  Systems},'' in {\em ACM SIGSAC Conference on Computer \&\#38; Communications
  Security}, pp.~901--914, 2013.

\bibitem{70}
T.~S. Ferguson, ``Linear programming: A concise introduction.''
\newblock Accessed: September 9, 2019.

\bibitem{55}
F.~{McSherry} and K.~{Talwar}, ``{Mechanism Design via Differential Privacy},''
  in {\em 48th Annual IEEE Symposium on Foundations of Computer Science},
  pp.~94--103, 2007.

\bibitem{54}
{Zhou Wang}, A.~C. {Bovik}, H.~R. {Sheikh}, and E.~P. {Simoncelli}, ``{Image
  Quality Assessment: From Error Visibility to Structural Similarity},'' {\em
  IEEE Transactions on Image Processing}, vol.~13, no.~4, pp.~600--612, 2004.

\bibitem{25}
M.~Flynn, ``{Generating Faces with Deconvolution Networks},'' 2016.
\newblock Accessed: November 1, 2018.

\bibitem{16}
O.~Langner, R.~Dotsch, G.~Bijlstra, D.~Wigboldus, S.~Hawk, and A.~van
  Knippenberg, ``{Presentation and Validation of the Radboud Faces Database},''
  {\em Cognition and Emotion}, vol.~24, no.~8, p.~1377—1388, 2010.

\bibitem{17}
D.~Lundqvist, A.~Flykt, and A.~{\"O}hman, ``{The Karolinska Directed Emotional
  Faces – KDEF},'' 1998.
\newblock ISBN 91-630-7164-9.

\bibitem{44}
K.~Simonyan and A.~Zisserman, ``{Very Deep Convolutional Networks for
  Large-Scale Image Recognition},'' in {\em 3rd International Conference on
  Learning Representations}, 2015.

\bibitem{45}
X.~Glorot and Y.~Bengio, ``{Understanding the Difficulty of Training Deep
  Feedforward Neural Networks},'' in {\em 13th International Conference on
  Artificial Intelligence and Statistics}, vol.~9, pp.~249--256, 2010.

\bibitem{67}
C.~B. Ng, Y.~H. Tay, and B.-M. Goi, ``{Recognizing Human Gender in Computer
  Vision: A Survey},'' in {\em PRICAI 2012: Trends in Artificial Intelligence},
  pp.~335--346, 2012.

\bibitem{21}
G.~Levi and T.~Hassner, ``{Age and Gender Classification Using Convolutional
  Neural Networks},'' in {\em IEEE Computer Vision and Pattern Recognition
  Workshops}, pp.~34--42, 2015.

\bibitem{75}
C.~Song and V.~Shmatikov, ``{Overlearning Reveals Sensitive Attributes},'' in
  {\em 8th International Conference on Learning Representations}, 2020.

\bibitem{64}
Y.-L. Tian, T.~Kanade, and J.~F. Cohn, ``{Facial Expression Analysis},'' in
  {\em {Handbook of Face Recognition}} (S.~Li and A.~Jain, eds.), ch.~11,
  pp.~247--275, Springer Science+Business Media Inc., 2005.

\bibitem{56}
D.~{Acharya}, Z.~{Huang}, D.~P. {Paudel}, and L.~{Van Gool}, ``{Covariance
  Pooling for Facial Expression Recognition},'' in {\em IEEE/CVF Conference on
  Computer Vision and Pattern Recognition Workshops}, pp.~480--4807, 2018.

\bibitem{65}
S.~Li, W.~Deng, and J.~Du, ``{Reliable Crowdsourcing and Deep
  Locality-Preserving Learning for Expression Recognition in the Wild},'' in
  {\em IEEE Conference on Computer Vision and Pattern Recognition}, vol.~28,
  pp.~2584--2593, 2017.

\bibitem{66}
A.~{Dhall}, R.~{Goecke}, S.~{Lucey}, and T.~{Gedeon}, ``{Static Facial
  Expression Analysis in Tough Conditions: Data, Evaluation Protocol and
  Benchmark},'' in {\em IEEE International Conference on Computer Vision
  Workshops}, pp.~2106--2112, 2011.

\bibitem{85}
H.~Ng and S.~Winkler, ``{A Data-Driven Approach to Cleaning Large Face
  Datasets},'' in {\em {IEEE} International Conference on Image Processing},
  pp.~343--347, 2014.

\bibitem{86}
F.~{Schroff}, D.~{Kalenichenko}, and J.~{Philbin}, ``{FaceNet: A Unified
  Embedding for Face Recognition and Clustering},'' in {\em 2015 IEEE
  Conference on Computer Vision and Pattern Recognition (CVPR)}, pp.~815--823,
  2015.

\bibitem{87}
M.~P. {Sampat}, Z.~{Wang}, S.~{Gupta}, A.~C. {Bovik}, and M.~K. {Markey},
  ``{Complex Wavelet Structural Similarity: A New Image Similarity Index},''
  {\em IEEE Transactions on Image Processing}, vol.~18, no.~11, pp.~2385--2401,
  2009.

\bibitem{88}
C.~{Szegedy}, {Wei Liu}, {Yangqing Jia}, P.~{Sermanet}, S.~{Reed},
  D.~{Anguelov}, D.~{Erhan}, V.~{Vanhoucke}, and A.~{Rabinovich}, ``{Going
  Deeper with Convolutions},'' in {\em 2015 IEEE Conference on Computer Vision
  and Pattern Recognition (CVPR)}, pp.~1--9, 2015.

\end{thebibliography}

\clearpage

\appendix

\section*{Appendix A - Facial Classification Network Architectures}

\begin{table} [!h]
	\centering
	\begin{tabular}{ | c | c | }
		\hline
		\textbf{Layer Type} & \textbf{Settings} \\ \hline
		Conv & Size: 3x3, Filters: 64, Stride: 1, Pad: 1 \\ \hline
		ReLU & \\ \hline
		Conv & Size: 3x3, Filters: 64, Stride: 1, Pad: 1 \\ \hline
		ReLU & \\ \hline
		Max Pool & Size: 2x2, Stride: 2, Pad: 0 \\ \hline
		Conv & Size: 3x3, Filters: 128, Stride: 1, Pad: 1 \\ \hline
		ReLU & \\ \hline
		Conv & Size: 3x3, Filters: 128, Stride: 1, Pad: 1 \\ \hline
		ReLU & \\ \hline
		Max Pool & Size: 2x2, Stride: 2, Pad: 0 \\ \hline
		Conv & Size: 3x3, Filters: 256, Stride: 1, Pad: 1 \\ \hline
		ReLU & \\ \hline
		Conv & Size: 3x3, Filters: 256, Stride: 1, Pad: 1 \\ \hline
		ReLU & \\ \hline
		Conv & Size: 3x3, Filters: 256, Stride: 1, Pad: 1 \\ \hline
		ReLU & \\ \hline
		Max Pool & Size: 2x2, Stride: 2, Pad: 0 \\ \hline
		Conv & Size: 3x3, Filters: 512, Stride: 1, Pad: 1 \\ \hline
		ReLU & \\ \hline
		Conv & Size: 3x3, Filters: 512, Stride: 1, Pad: 1 \\ \hline
		ReLU & \\ \hline
		Conv & Size: 3x3, Filters: 512, Stride: 1, Pad: 1 \\ \hline
		ReLU & \\ \hline
		Max Pool & Size: 2x2, Stride: 2, Pad: 0 \\ \hline
		Conv & Size: 3x3, Filters: 512, Stride: 1, Pad: 1 \\ \hline
		ReLU & \\ \hline
		Conv & Size: 3x3, Filters: 512, Stride: 1, Pad: 1 \\ \hline
		ReLU & \\ \hline
		Conv & Size: 3x3, Filters: 512, Stride: 1, Pad: 1 \\ \hline
		ReLU & \\ \hline
		Max Pool & Size: 2x2, Stride: 2, Pad: 0 \\ \hline
		FC & Size: 7x7, Filters: 4096, Stride: 1, Pad: 0 \\ \hline
		ReLU & \\ \hline
		Dropout & Rate: 0.5 \\ \hline
		FC & Size: 1x1, Filters: 4096, Stride: 1, Pad: 0 \\ \hline
		ReLU & \\ \hline
		Dropout & Rate: 0.5 \\ \hline
		FC & Size: 1x1, Filters: IDs, Stride: 1, Pad: 0 \\ \hline
		Softmax & \\ \hline
	\end{tabular}
	\caption{VGG D architecture \cite{19}. We set the number of filters on the final fully connected layer to the number of identities in the dataset on which the network is applied.}
	\label{tab:vgg}
\end{table}

\begin{table} [!h]
	\centering
	\begin{tabular}{ | c | c | }
		\hline
		\textbf{Layer Type} & \textbf{Settings} \\ \hline
		Conv & Kernel: 7x7, Filters: 64, Stride: 2 \\ \hline
		ReLU & \\ \hline
		Max Pool & Size: 3x3, Stride: 2 \\ \hline
		Inception & \\ \hline
		Max Pool & Size: 3x3, Stride: 2 \\ \hline
		Inception & Uses L2 Pooling \\ \hline
		Inception & Uses L2 Pooling \\ \hline
		Inception & \\ \hline
		Max Pool & Size: 3x3, Stride: 2 \\ \hline
		Inception & Uses L2 Pooling \\ \hline
		Inception & Uses L2 Pooling \\ \hline
		Inception & Uses L2 Pooling \\ \hline
		Inception & Uses L2 Pooling \\ \hline
		Inception & \\ \hline
		Max Pool & Size: 3x3, Stride: 2 \\ \hline
		Inception & Uses L2 Pooling \\ \hline
		Inception & \\ \hline
		Average Pool & Size: 7x7, Stride: 1 \\ \hline
		FC & Size: 1x1, Filters: 128, Stride: 1\\ \hline
		L2 Norm & \\ \hline
	\end{tabular}
	\caption{FaceNet architecture \cite{86}. Inception refers to the use of inception blocks as described in \cite{88}. Inception blocks listed as using L2 pooling do so in place of the standard max pooling.}
	\label{tab:facenet}
\end{table}

\newpage

\begin{figure*}
	\section*{Appendix B - Additional Pixel-Space MSE Results}
	
	\centering
	\includegraphics[width=0.95\textwidth]{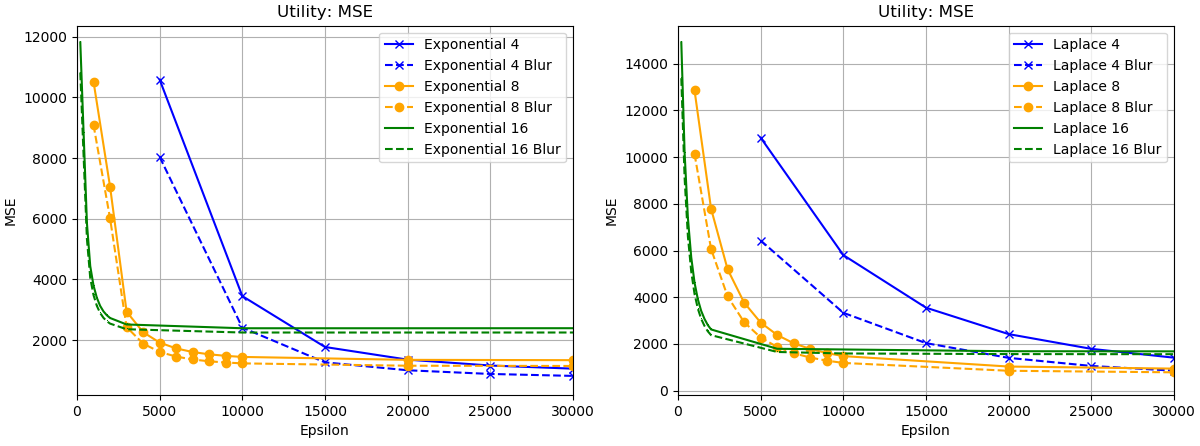}
	\caption{A comparison of the MSE achieved by the mechanisms over a wide range of privacy budgets. The left graph compares variants of the exponential mechanism and the right graph compares variants of the Laplace mechanism.}
	\label{fig:epsilon_mse}
\end{figure*}

\begin{figure*}
	\centering
	\includegraphics[width=0.95\textwidth]{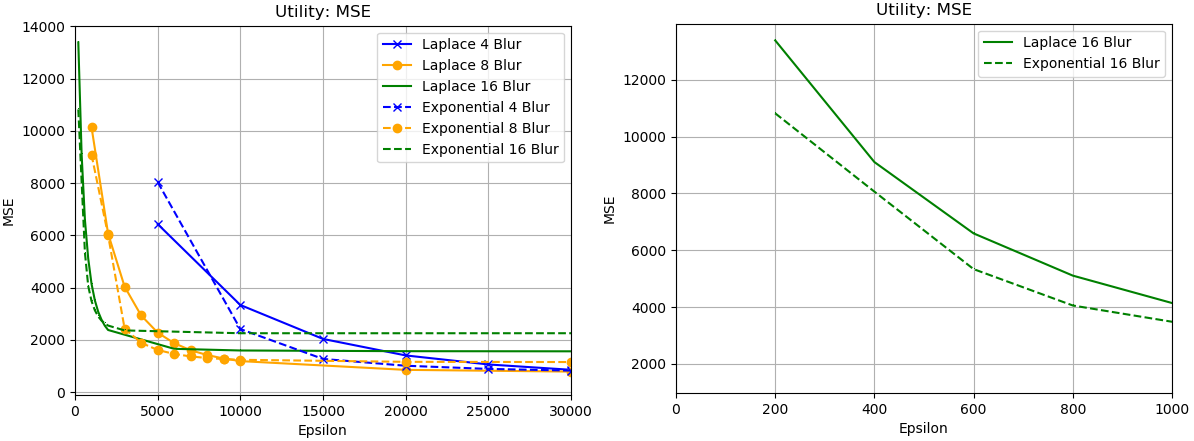}
	\caption{A comparison between the blur variants of the exponential and Laplace mechanism with respect to utility measured using MSE. The left graph compares all three pixelization settings for the mechanisms and the right graph focuses on the mechanisms using a pixelization grid of size 16 for efficient use of low privacy budgets.}
	\label{fig:epsilon_mse_both}
\end{figure*}

\end{document}